\begin{filecontents*}{}
\documentclass[smallextended]{svjour3}       
\smartqed  
\usepackage{graphicx}
\usepackage{amsmath}
\usepackage[linesnumbered]{algorithm2e}
\usepackage{color}

\newcommand{\old}[1]{{}}

\def\eps{\varepsilon}
\def\floor#1{\left\lfloor#1\right\rfloor}
\def\ceil#1{\left\lceil#1\right\rceil}

\begin{document}

\title{Geometric Hitting Set for Segments of Few Orientations}


\author{
S\'andor P. Fekete \and 
Kan Huang \and \\
Joseph S. B. Mitchell \and
Ojas Parekh \and
Cynthia A. Phillips 
}

\authorrunning{S.~P.~Fekete, K.~Huang, J.~S.~B.~Mitchell, O.~Parekh and C.~A.~Phillips}  

\institute{
S.~P.~Fekete \at
TU Braunschweig, Braunschweig, Germany \\
\email{s.fekete@tu-bs.de}
\and
K. Huang \at
Stony Brook University, Stony Brook, NY, USA \\
\email{khuang@ams.stonybrook.edu}
\and 
J.~S.~B.~Mitchell \at
Stony Brook University, Stony Brook, NY, USA \\
\email{jsbm@ams.stonybrook.edu}
\and
O.~Parekh \at 
Sandia National Labs, Albuquerque, NM, USA \\
\email{odparek@sandia.gov}
\and 
C.~A.~Phillips \at
Sandia National Labs, Albuquerque, NM, USA \\
\email{caphill@sandia.gov}
}

\date{}

\maketitle

\begin{abstract}
We study several natural instances of the geometric hitting set
problem for input consisting of sets of line segments (and rays,
lines) having a small number of distinct slopes.  These problems
model path monitoring (e.g., on road
networks) using the fewest sensors (the ``hitting points'').
We give approximation algorithms for cases including
(i) lines of 3
slopes in the plane, (ii) vertical lines and horizontal segments,
(iii) pairs of horizontal/vertical segments.  We give
hardness and hardness of approximation results for these problems.  We
prove that the hitting set problem for vertical lines and horizontal
rays is polynomially solvable.
\keywords{set cover \and hitting set \and approximation algorithms}
\end{abstract}

\section{Introduction} \label{sec:intro}

A fundamental problem in combinatorial optimization is the {\em set cover problem},
in which we are given a collection, ${\cal C}$,
of subsets of a set $U$, of elements, and our goal is to find a minimum-cardinality subset of
${\cal C}$ whose union covers $U$.
The set cover problem is NP-hard and has an $O(\log
n)$-approximation algorithm, which is 
best possible in the worst case
(unless $P=NP$, ~\cite{dinur2014analytical}).  
Equivalently, set cover can be     
cast as a {\em hitting set problem}: given a collection, ${\cal C}$,
of subsets of set $U$, find a smallest cardinality set $H\subseteq U$
such that every set in ${\cal C}$ contains at least one element of $H$.
Numerous special instances of set cover/hitting set have been studied.
Our focus in this paper is on geometric instances that arise in
covering (hitting) sets of (possibly overlapping) line segments using
the fewest points (``hit points'').  A closely related problem is
the ``Guarding a Set of Segments'' (GSS)
problem~\cite{brimkov2013approximability,brimkov2011guarding,brimkov2012approximation,joshi2014approximation}, in which
the segments may cross arbitrarily, but do not
overlap.  Since this problem is strongly
NP-complete~\cite{brimkov2011guarding} in general, our focus is on
special cases, primarily those in which the segments come from a small
number of orientations (e.g., horizontal, vertical).  We provide several new
results on hardness and approximation algorithms.

We are motivated by the path monitoring problem: 
given a set of trajectories, each a path of line segments in the plane,
place the fewest sensors (points) to observe (hit) all trajectories.
To gain theoretical insight into this challenging problem, we examine
cleaner, but progressively harder, versions of hitting
trajectory/line-like objects with points.  If the trajectories are
on a Manhattan road network, the paths are (possibly
overlapping) horizontal/vertical segments.  
Alternatively, one wishes to place the fewest vendors or
service stations in a road network to service a set of customer trajectories.

\paragraph{Our results.} 
We give complexity and approximation results for several geometric hitting set problems on inputs $S$ of line
``segments'' of special classes, mostly of fixed orientations.  The
segments are allowed to overlap arbitrarily. We consider various cases
of ``segments'' that may be bounded (line segments), semi-infinite
(rays), or unbounded in both directions (lines).  Our results
are:

(1) Hitting lines of 3 slopes in the plane is NP-hard (greedy is optimal for 2 slopes).
For set cover with set size at most 3,
standard analysis of the greedy
algorithm gives an approximation factor of $H(3)=1+(1/2)+(1/3)=(11/6)$, and there is a 4/3-approximation based on semi-local optimization~\cite{duh1997approximation}.  We prove that the greedy algorithm in this special geometric case is a ($7/5$)-approximation.

(2) Hitting vertical lines and horizontal rays is polytime solvable.

(3) Hitting vertical lines and horizontal (even unit-length) segments
is NP-hard.
Our proof shows hitting horizontal and
vertical unit-length segments is also NP-hard.
 We prove APX-hardness for hitting horizontal and vertical segments.

(4) Hitting vertical lines and horizontal segments has a
$(5/3)$-approximation algorithm.  (This problem has a straightforward 2-approximation.)

(5) Hitting pairs of horizontal/vertical segments has a
4-approximation.  Hitting pairs having one vertical and one
horizontal segment has a ($10/3$)-approx\-imation.  These results are based
on LP-rounding.  More generally, hitting sets
of $k$ segments from $r$ orientations has a
$(k \cdot r)$-approximation algorithm.

(6) We give a linear-time combinatorial 3-approximation algorithm
for hitting triangle-free
sets of (non-overlapping) segments.  Recently Joshi and Narayanaswamy~\cite{joshi2014approximation} gave
a 3-approx\-imation for this version of GSS 
using linear programming.

\paragraph{Related Work.}
There is a wealth of related work on geometric set cover and hitting
set problems; we do not attempt here to give an exhaustive survey.
The {\em point line cover} (PLC) problem (see
\cite{heednacram2010np,kratsch2014point}) asks for a smallest set of
lines to cover a given set of points; it is equivalent, via point-line
duality, to the hitting problem for a set of lines.  The PLC (and thus
the hitting problem for lines) was shown to be
NP-hard~\cite{megiddo1982complexity}; in fact, it is
APX-hard~\cite{DBLP:conf/cccg/BrodenHN01} and Max-SNP
Hard~\cite{kumar2000hardness}.  The problem has an $O(\log
OPT)$-approximation (e.g., greedy -- see
\cite{kovaleva2006approximation}); in fact, the greedy algorithm for
PLC has worst-case performance ratio $\Omega(\log
n)$~\cite{DBLP:journals/comgeo/DumitrescuJ15}.  Afshani et
al.~\cite{DBLP:conf/compgeom/AfshaniBDN16} have studied exact and
parameterized algorithms, giving an $O^*(2^n)$ time algorithm that
uses polynomial space, and an $O^*((Ck/\log k)^{(d-1)k})$ time
algorithm to hit $n$ lines with the minimum number of points or at
most $k$ points.  (Here, $O^*(\cdot)$ indicates that polynomial
factors of $n$ are hidden.)

Hassin and Megiddo~\cite{hassin1991approximation} considered
hitting geometric objects with the fewest
lines having a small number of distinct slopes.  They observed that,
even for covering with axis-parallel lines, the greedy algorithm has
an approximation ratio that grows logarithmically.  They gave
approximations for the problem of hitting horizontal/vertical segments
with the fewest axis-parallel lines (and, more generally, with lines
of a few slopes).
Gaur and Bhattacharya~\cite{gaur2007covering} consider covering points
with axis-parallel lines in $d$-dimensions. They give a
$(d-1)$-approximation based on rounding the corresponding linear
program (LP).  Many other stabbing problems (find a small
set of lines that stab a given set of objects) have been studied; see,
e.g.,
\cite{dom2009parameterized,even2008algorithms,gaur2000constant,giannopoulos2013fixed,kovaleva2006approximation,langerman2005covering}.

A recent paper~\cite{joshi2014approximation} gives a 3-approximation
for hitting sets of ``triangle-free''segments. Brimkov et
al.~\cite{brimkov2013approximability,brimkov2011guarding,brimkov2012approximation}
have studied the hitting set problem on line segments, including
various special cases; they refer to the problem as ``Guarding a Set
of Segments,'' or GSS. GSS is a special case of the ``art
gallery problem:'' place a small number of
``guards'' (e.g., points) so that every point within a geometric
domain is ``seen'' by at least one
guard~\cite{o-agta-87,urrutia2000art}.  Brimkov et
al.~\cite{brimkov2010experimental} provide experimental results for
three GSS heuristics, including two variants of ``greedy,''
showing that in practice the algorithms perform well and are often
optimal or very close to optimal.  They prove, however, that, in
theory, the methods do not provide worst-case constant-factor
approximation bounds.  For the special case that the segments are
``almost tree (1)'' (a connected graph is an {\em almost tree ($k$)}
if each biconnected component has at most $k$ edges not in a spanning
tree of the component), a $(2-\eps)$-approximation is
known~\cite{brimkov2013approximability}.

An important distinction between GSS and our problems
is that we allow {\em overlapping} (or partially overlapping) segments (rays, and lines),
while, in GSS, each line segment is maximal in the
input set of line segments (the union of two distinct input segments
is not a segment).  A special case of our problem is
{\em interval stabbing} on a line: Given a set of segments
(intervals), arbitrarily overlapping on a line, find a smallest
hitting set of points that hit all segments.  A simple sweep along
the line solves this problem optimally: when a segment ends, place
a point and remove all segments covered by that point.  

If no point lies within three or more objects, then the hitting
set problem is an edge cover problem in the
intersection graph of the objects.  In particular, if no three
segments pass through a common point, the problem can be solved
optimally in polynomial time.
(This implies that in an arrangement of ``random'' segments, the GSS
problem is almost surely polynomially solvable; see
\cite{brimkov2013approximability}.)

Hitting axis-aligned rectangles is related to
hitting horizontal and vertical segments.
Aronov, Ezra, and Sharir~\cite{aronov2010small} provide an $O(\log\log
OPT)$-approximation for hitting set for axis-aligned rectangles (and
axis-aligned boxes in 3D), by proving a bound of
$O(\eps^{-1} \log\log (\eps^{-1}))$ on the $\eps$-net size
of the corresponding range space.
The connection between hitting sets and
$\eps$-nets~\cite{bg-aoscf-95,clarkson1993algorithms,cv-iaagsc-07,even2005hitting}
implies a $c$-approximation for hitting set if one can compute an
$\eps$-net of size $c/\eps$; recent major advances~\cite{alon2012non,pach2013tight} on lower bounds on
$\eps$-nets imply that associated range spaces (rectangles and points,
lines and points, points and rectangles) have $\eps$-nets of size
superlinear in $1/\eps$.
Remarkably, improved $(1+\eps)$-approximation algorithms
(i.e., PTASs) for certain geometric hitting set and set cover
problems are possible with simple local search. For example, Mustafa
and Ray~\cite{mustafa2010improved} give a local search
PTAS for computing a smallest subset, of a given set of disks, that
covers a given set of points.  Hochbaum and Maas~\cite{hm-ascpp-85}
used grid shifting to obtain a much earlier PTAS for the
minimum unit disk cover problem when disks can be placed
anywhere in the plane, not restricted to a discrete input set.

\section{Hitting Segments}
\label{sec:seg}
Suppose $S$ is a set of $n$ line segments in the plane.
If all segments are horizontal, then we can compute an optimal hitting
set by independently solving the interval stabbing problem
along each of the horizontal lines determined by the input.
The time required is $O(n\log n)$, used to sort the segment endpoints along their containing line(s).

If the segments are of two different orientations (slopes), then the
problem becomes significantly harder.  
By applying an affine transformation to the segments $S$ if necessary, we can, without loss of generality,
assume the segments are horizontal and vertical.
We show the problem is hard even if the axis-parallel segments are all the
same length. This result (Corollary~\ref{cor:unit_npc}) is
a consequence of an even stronger result, Theorem~\ref{th:rays},
which we establish in Section~\ref{sec:mix}.

By solving optimally
each of the two (or $k$) orientations, and using the union of the hitting
points for both (or all $k$), we obtain:

\begin{theorem}
For a set $S$ of $n$ line segments having $k$ different orientations (slopes) in the plane,
we can compute, in time $O(n\log n)$, a $k$-approximation for the optimal hitting set.
\end{theorem}

\section{Hitting Lines}
\label{sec:lin}
When  $S$ is a set of $n$ lines in the
plane, greedy gives an $O(\log OPT)$ approximation
factor; any
approximation factor better than logarithmic would be quite
interesting.  (See \cite{DBLP:journals/comgeo/DumitrescuJ15,kratsch2014point}.)
If the lines have only 2 slopes, then an optimal algorithm is given by 
the greedy selection of hitting points: Add to the hitting set
(initially empty) any point at the intersection of two unhit lines; if
no such point exists, and there are still unhit lines, then add to the
hitting set a point on an unhit line.  
(For, say, $n_h$ horizontal lines and $n_v\geq n_h$ horizontal lines,
the greedy algorithm uses $n_v$ hitting points, first selecting $n_h$
points that each hit two previously unhit lines (one horizontal, one
vertical), in any order, then selecting $n_v-n_h$ additional hitting
points, in any order, each hitting a single unhit vertical line.)

\subsection{Hardness of Hitting Lines of 3 Slopes in 2D}

We prove that the hitting set problem is NP-hard when the set $S$ of
input lines have more than two slopes.  In particular, we show below that
the problem, {\sc 3-Slope-Line-Cover} (3SLC), of computing a
minimum-cardinality hitting set for a set $S$ of lines having three
distinct slopes is NP-hard.  (The corresponding decision problem is NP-complete.)

\begin{theorem}
\label{h3ls}
The problem 3SLC is NP-complete.
\end{theorem}

\begin{proof}
For convenience, we recast the 3SLC problem
into its equivalent dual formulation: Find a minimum-cardinality set
of non-vertical lines to cover a set $P$ of points (duals to the set
$S$ of lines), which are known to lie on three vertical lines.  Here,
we are using the notion of ``point-line'' duality, in which a point
$p=(a,b)$ in the ``primal plane'' has a corresponding (non-vertical)
line, $p^*: y=ax-b$, in the ``dual plane,'' and a non-vertical line
$L: y=mx+b$ in the ``primal plane'' has a corresponding point,
$L^*=(m,-b)$, in the ``dual plane.''  Point-line duality is a
one-to-one mapping between points and non-vertical lines, preserving
incidence and order; any statement about points and lines is mapped,
via duality, to an equivalent statement about lines and points.  Since
the lines of $S$ in our 3SLC instance have three slopes, their dual
points $P$ have $x$-coordinates among three possible values and thus lie
on three vertical lines.  (See \cite{bcko} for background and
applications of point-line duality in computational geometry.)
Our reduction is from 3SAT.  Let $n$ and $m$ denote the numbers of
variables and clauses respectively.  From an instance of 3SAT we
create an instance, $P$, of the (dual formulation) of a 3SLC problem,
as follows.  The points $P$ are distributed on three vertical lines,
denoted $l_1$, $l_2$ and $l_3$, from left to right.

We use the following terminology.  If a line $l$ covers $i$ points, we say that $l$ is an {\em
  $i$-line}. Let $P_l$ denote the set of points of $P$ that are
  covered by line $l$.  If $P_{l_1} \cap P_{l_2}
  = \emptyset$, then we say that lines $l_1$ and $l_2$ are {\em
  independent}.  A set $L$ of lines is {\em independent} if the lines
  are pairwise independent; i.e., $l_i\in L$ and $l_j\in L$ are
  independent for every $l_i\neq l_j$.

A variable gadget for a variable $u$ has $m$ points ($a[1],
a[2],\ldots,a[m]$) on line $l_1$, $m$ points ($b[1],
b[2], \ldots,b[m]$) on line $l_2$ and $2m$ points ($u[1],
u[2], \ldots, u[m]$, $\bar u[1], \bar u[2], \ldots, \bar u[m]$),
corresponding to the variables and their negations, on line $l_3$.
The points are placed so that through each point $b[i]$ on line $l_2$
there are exactly two 3-lines, a ``red'' line passing through $b[i]$
and $u[i]$ and one of the points $a[j]$ on line $l_1$, and a ``blue''
line passing through $b[i]$ and $\bar u[i]$ and one of the points
$a[j]$ on line $l_1$.  Figure~\ref{fig:vg} shows a variable gadget for
variable $u$ in a $4$-clause ($m=4$) instance. Points on the two lines
$l_1$ and $l_2$ can be hit by either of two sets of independent lines,
the set of red lines or the set of blue lines. These represent the
``True'' or ``False'' setting of the variable~$u$ respectively.  We
add the variable gadgets one by one onto the three lines $l_1$, $l_2$,
and $l_3$ so that all 3-lines are within variable gadgets (there are
no other triples of of colinear points, with one per vertical line).

The clause gadgets link variable gadgets.  The $i$th clause has a
single point $c[i]$ on line $l_1$. Point $c[i]$ is connected by three
``green'' line segments to the negations of the literals in its clause on line
$l_3$. Figure~\ref{fig:cg} shows the gadget for clause $c[i]
= \bar{u} \vee v \vee \bar{w}$.  We then include three additional
points in the clause gadget, namely the three points, $d[i,1]$,
$d[i,2]$, and $d[i,3]$, where these three line segments incident on
$c[i]$ cross line $l_2$.
We pick the locations of the points $c[i]$ so that the addition of
these four new points ($c[i], d[i,1], d[i,2], d[i,3]$) create no new
3-lines other than the ``green'' lines associated with the clause
gadget.  (This is easily done, since the creation of a new 3-line
would require that $c[i]$ be placed at one of a discrete set of
possible locations on $l_1$; since there is a continuum of possible
locations on $l_1$, the discrete locations can be avoided.)

To complete the construction, we place $nm + m$ points on line $l_1$
such that these new points are not on any $3$-line.  Thus, by
construction, each $3$-line is in a variable gadget or a clause
gadget.  There are $2nm + 2m$ points on line $l_1$, $nm + 3m$ points
on line $l_2$, and $2nm$ points on line $l_3$ for a total of $5nm +
5m$ points on all three lines.

We now argue that the 3SAT formula is satisfiable if and only if the
corresponding (dual formulation) 3SLC instance can be covered by
$2nm+2m$ non-vertical lines.

If the 3SAT formula is satisfiable, there is an independent set of
3-lines in the variable and clause gadgets of size $nm+m$.
Specfically, there are $nm$ independent $3$-lines corresponding to the
truth assignment for the $n$ variables. Since the formula is
satisfiable, for each clause, the $3$-line corresponding to the
negation of one correctly set variable is independent of the variable
truth assignment points, and can be part of an independent set.  This
gives the other $m$ members of the independent set.  There are then
$nm+m$ points on $l_1$, $2m$ points on $l_2$ and $nm-m$ points on
$l_3$ not part of this maximum independent set.  Each of the $nm + m$
remaining points on line $l_1$ can be paired with one of the remaining
$nm + m$ points on either line $l_2$ or $l_3$.  Thus, all points are
covered with $(nm + m)$ $3$-lines and $(nm + m)$ $2$-lines for a total
of $2nm + 2m$ lines.

If the instance can be covered with $2nm + 2m$ lines, then, by
construction, $nm + m$ of these lines must be $2$-lines involving the
points on line $l_1$ that are not in any clause or variable gadgets.
These $2$-lines can cover $2nm + 2m$ points, leaving $3nm + 3m$ total
points to be covered by the remaining $nm + m$ lines. Thus, the
remaining lines must all be independent $3$-lines.  $nm$ of these
correspond to a truth assignment for the variables.  The remaining $m$
must come from clause gadgets.  Thus, each clause gadget has a
$3$-line compatible with the truth assignment and the 3SAT instance is
satisfiable.
\qed
\end{proof}

\begin{figure}[htpb]
\centering
\begin{minipage}[htb]{0.45\linewidth}
  \centering
  \includegraphics[scale=0.5]{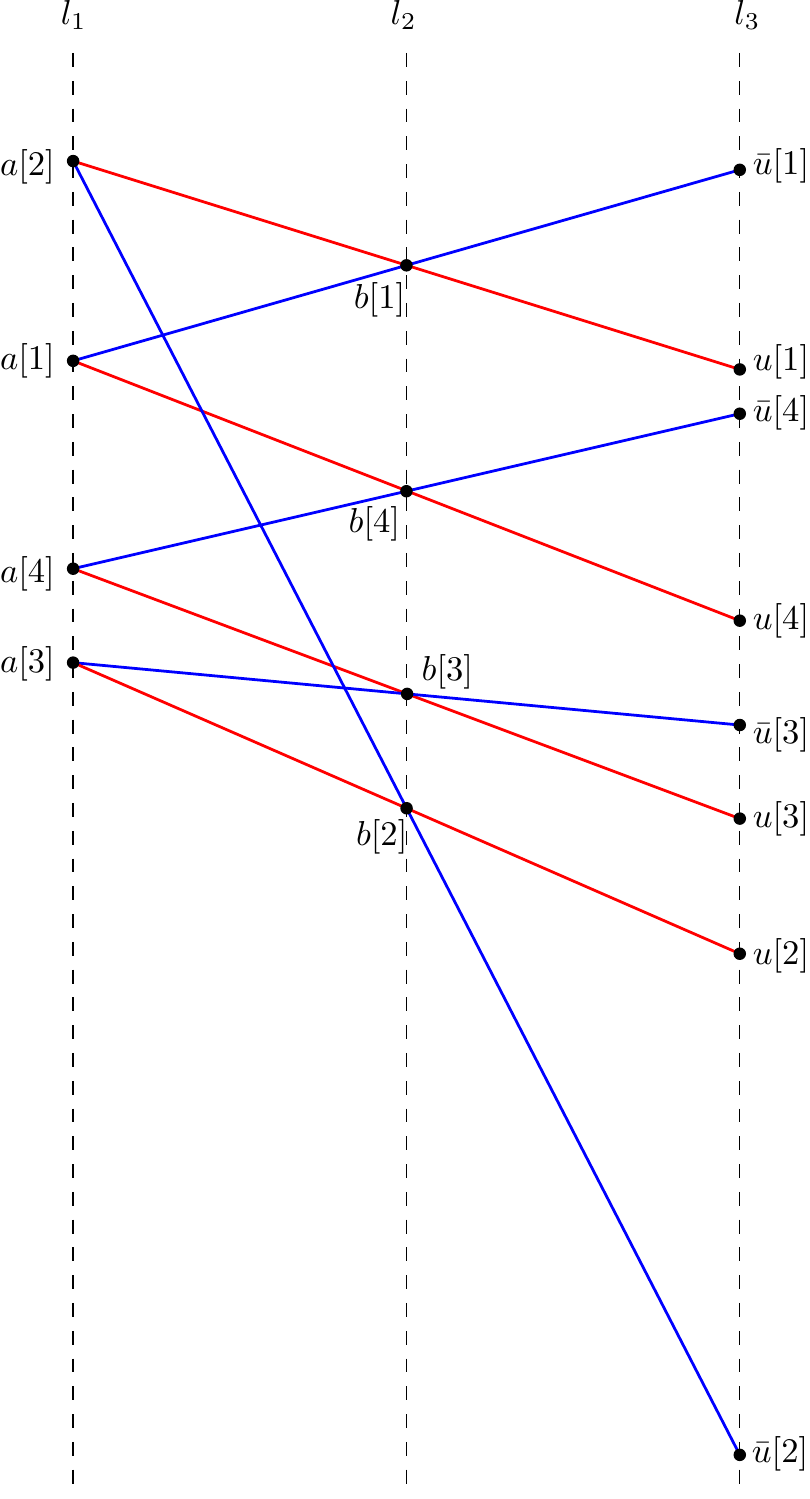}
\caption{A variable gadget.}
\label{fig:vg}
\end{minipage}
\quad
\begin{minipage}[htb]{0.45\linewidth}
  \centering
  \includegraphics[scale=0.6]{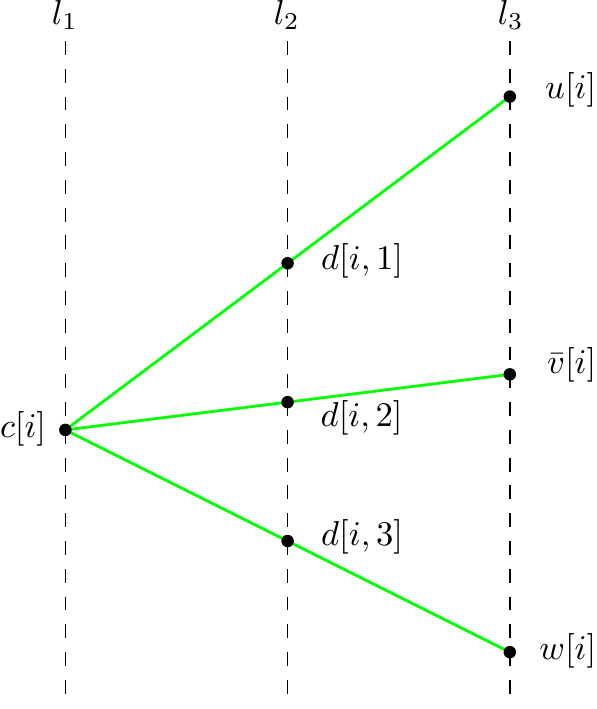}
\caption{A clause gadget.}
\label{fig:cg}
\end{minipage}

\end{figure}

\subsection{Analysis of the Greedy Hitting Set Algorithm for Lines of 3 Slopes in 2D}

If no point lies in more than $k$ sets, the greedy algorithm's
approximation factor is $H(k) = \sum_{i=1}^k
(1/i)$~\cite{chvatal1979greedy}. 
This property holds for lines of $3$ slopes with $k=3$, giving a greedy
approximation factor $H(3)=11/6$.  
We give a new analysis, exploiting the special geometric
structure of the hitting set problem for lines of 3 slopes, to obtain
an approximation factor of $7/5$.

Let $x$, $y$ and $z$ be the number of lines in each of the three
slopes in the plane. Without loss of generality, we assume that
$x\geq y \geq z>0$.

We call a point where at least two lines meet a vertex.  A vertex is a
3-intersection if three lines, one from each orientation, meet at that
point.  Otherwise it is a 2-intersection.  If there are no
3-intersections, then all vertices are 2-intersections.  We claim then
that the greedy algorithm is optimal, with the following specification
of how to break ties: As long as there is a vertex hitting two unhit
lines, pick any one of them that hits the two most populous sets
(slopes) of unhit lines, remove the newly hit lines, and repeat.
Since, in our notation, we assume that $x\geq y \geq z>0$, this means
that, with the selection of the next vertex, $x$ and $y$ each go down
by exactly one; we then update the labels $x$, $y$, and $z$ (since it
could be that the population, $z$, of the third slope category is now
greater than one or both of the updated populations, $x-1$ and $y-1$),
and repeat the process.  Once there are lines of at most one slope
(i.e., $x\geq 0$, $y=z=0$), the greedy algorithm is forced to place one
hit point along each of the $x$ lines.
Let $OPT_2(x,y,z)$ denote the minimum number of hit points for a
3-slope instance with no 3-intersections, and $x\geq y\geq z$ lines of
each of the three distinct slopes.

\begin{lemma}
\label{lem:2-intersections-only}
The above greedy algorithm yields an optimal hitting set when applied
to an input set of lines of 3 slopes no three of which pass through a
common point.  Further, the optimal number of hit points is given by
$$
OPT_2(x,y,z) = \left\{ \begin{array}{ll}
x & \mbox{  if $x \geq y+z$}; \\
x + \ceil{\frac{y+z-x}{2}} = \ceil{\frac{x+y+z}{2}} & \mbox{  if $x < y+z$}.
\end{array}
\right.
$$
\end{lemma}

\begin{proof}
First, note that an optimal hitting set must use at least $x$ hit
points, since each of the $x$ lines in the most populous slope class
must be hit by a distinct hit point.  Also, since no point hits more
than two lines, we know that an optimal hitting set must use at least
$\ceil{(x+y+z)/2}$ hit points.

If $x\geq y+z$, then $x$ hit points suffice and are optimal: simply
place one hit point on each of the $x$ lines of the first slope class:
for the first $y+z$ (of the $x$) lines in the class, place the hit
points at $y+z$ vertices where the lines of the first slope class
cross lines of the second and third slope classes, thereby hitting all
$y+z$ such lines; once there are only lines of the first slope class
left, place exactly one hit point on each of the remaining $x-(y+z)$
lines of the first slope class.  This results in an optimal hitting
set with $x$ hit points.  The greedy algorithm will produce an optimal
such set, since it will continue to place hit points at crossing
points of the first and second slope class, until the population of
the second class drops to one below that of the third (at which point
the class labels swap), and then continues (with the populations of
the second and third slope classes alternating in which one is
``$y$''), always able to hit a new line of the first slope class,
together with a new line of one of the other slope classes.  The
greedy algorithm continues in this way, placing hit points at
2-intersections, until the second and third slope classes are empty,
and the remaining $x-(y+z)$ hit points must all go on unhit lines of
the first slope class.

If $x\geq y+z$, then $\ceil{(x+y+z)/2}$ hit points suffice and are
optimal: one can always place a hit point at a 2-intersection where
two unhit lines cross, and we claim that this is, in fact, what the
greedy algorithm does.  We argue that at the first moment when there
are no 2-intersections that hit two unhit lines, the population vector
$(x,y,z)$ must be either (1,0,0) or (0,0,0).  To see this, consider
the three stages of the greedy algorithm:
\begin{description}
\item[(i)] 
$x$ and $y$ each go down by 1, $z$ is unchanged; this stage continues
until the population $y$ drops to the value $z-1$, causing the class
labels to swap, since now there are fewer lines in the second slope
class than in the third slope class; stage (ii) takes over.
\item[(ii)] 
$x$ and $y$ each go down by 1, $z$ is unchanged, but now, since $y$
and $z$ are within 1 of each other, the class labels swap back and
forth between the second and third slope class; this stage continues
until the population $x$ drops to within 1 of each of the other two
smaller populations, $y$ and $z$.
\item[(iii)] 
$x$ and $y$ each go down by 1, $z$ is unchanged, but now, since all
three populations $x$, $y$, and $z$, are within 1 of each other, class
labels swap around and the populations stay within 1 of each other,
until finally we reach population vector (1,0,0) or (0,0,0), and there
are no 2-intersection points of unhit lines.
\end{description}
It is easy to check that when $x\geq y+z$, $\ceil{(x+y+z)/2}=x + \ceil{(y+z-x)/2}$.
\qed
\end{proof}

When there are 3-intersections, the greedy algorithm iteratively
places hit points at 3-intersections (in any order), removing the 3
hits lines, until there are no more 3-intersections.  The remaining
instance, with only 2-intersections, now can now be solved optimally
as described above.

\begin{theorem}
The greedy algorithm yields a $\frac{7}{5}$-approximation for hitting lines of 3 slopes in 2D.
\label{theorem:3line_greedy}
\end{theorem}

\begin{proof}
Consider the graph $G$ whose vertices are the 3-intersections of the
input set of lines, with an edge of $G$ between two vertices if and
only if the corresponding 3-intersections lie on a common input line.  Let
$K$ denote the cardinality of a maximum independent set, $I_{max}$, in
$G$.  Let $N_3$ be the number of 3-intersections that the greedy
algorithm selects.  These $N_3$ vertices correspond to a maximal independent set in $G$, so $K \ge N_3$.
We have, $K\leq 3N_3$, since (by independence) there is at most one vertex of
$I_{max}$ along each of the $3N_3$ lines that the greedy algorithm hits with 3-intersections.

The optimal solution has $OPT$ hit points, with
\begin{align}
   OPT = K + OPT_2(x-K, y-K, z-K)
   \label{}
\end{align}
The greedy solution yields a set of $N_{greedy}$ hit points, with
\begin{align}
   N_{greedy} = N_3 + OPT_2(x-N_3, y-N_3, z-N_3).
   \label{}
\end{align}

First suppose $x=K$, which means that $x=y=z=K$. Then, 
$OPT=K$ and $N_{greedy}= N_3 + OPT_2(K-N_3, K-N_3, K-N_3) = K + \ceil{\frac{K-N_3}{2}}$.
We have three cases, depending on the value of $K$, mod 3:
\begin{enumerate}
    \item $K=3l$, where $l$ is an integer; thus, 
      \begin{align}
        \frac{N_{greedy}}{OPT}\leq 1 + \ceil{\frac{3l-l}{2}}/3l = \frac{4}{3}.
        \label{}
      \end{align}
    \item $K=3l+1$; thus, 
      \begin{align}
        \frac{N_{greedy}}{OPT}\leq 1 + \ceil{\frac{3l+1-(l+1)}{2}}/(3l+1) < \frac{4}{3}.
        \label{}
      \end{align}
    \item $K=3l+2$; thus, 
      \begin{align}
        \frac{N_{greedy}}{OPT} \leq 1 + \ceil{\frac{3l+2 - (l+1)}{2}}/(3l+2) \leq \frac{7}{5}.
        \label{}
      \end{align}
  \end{enumerate}

Now, suppose $x \geq K+1$. Then, we have the following three cases:
  \begin{enumerate}
    \item If $x \geq y + z - N_3$, then $x-N_3 \geq y-N_3 + z-N_3$, and $x-K \geq y-K + z-K$. Thus, 
      \begin{align*}
        OPT = K + x - K = x, \\
        N_{greedy} = N_3 + x - N_3 = x.\\
        \label{}
      \end{align*}

    \item If $y+z-K \leq x < y+z-N_3$, then 
      \begin{align}
        OPT &= x, \nonumber\\
        N_{greedy} &= x + \ceil{\frac{y+z-x-N_3}{2}},\label{Greedy-3}\\
        \frac{N_{greedy}}{OPT} &= 1 + \ceil{\frac{y+z-x-N_3}{2}}/x \leq 1 + \ceil{\frac{K-N_3}{2}}/(K+1) \leq \frac{4}{3}.\nonumber 
      \end{align}
The detailed analysis is similar to the case in which $x=K$. 

   \item If $x< y+z - K$, then 
     \begin{align}
       OPT &= x+ \ceil{\frac{y+z-x-K}{2}} \geq K + 1 + 1 = K+2, \label{OPT-3}\\
       N_{greedy} - OPT & \leq \frac{y+z-x-N_3}{2} + 1 - \frac{y+z-x-K}{2} \leq 1+ \frac{K}{3},\nonumber \\
       \frac{N_{greedy}}{OPT} & = 1 + \frac{N_{greedy}- OPT}{OPT} \leq 1+ \frac{1+K/3}{K+2}.\nonumber
       \label{}
     \end{align}
Thus, 
     \begin{itemize}
       \item when $K=1$, $OPT= N_{greedy}$; 
       \item when $K = 2, 3, 4$, using~(\ref{Greedy-3}) and the first equality in~(\ref{OPT-3}), we have $N_{greedy} - OPT = 1$, so we have
         \begin{align}
           \frac{N_{greedy}}{OPT} \leq 1 + \frac{1}{4} = 1.25;
           \label{}
         \end{align}
       \item when $K\geq 5$, 
         \begin{align*}
           \frac{N_{greedy}}{OPT}\leq 1 + \frac{8}{21} \approx 1.381.
         \end{align*}
     \end{itemize}
  \end{enumerate}
This completes the proof.
\qed
\end{proof}

\subsection{Axis-Parallel Lines in 3D}

While in 2D the hitting set problem for axis-parallel lines is easily
solved, in 3D we prove 
that the corresponding
hitting set problem is NP-hard, using a reduction from 3SAT.

\begin{theorem}
Hitting set for axis-parallel lines in 3D is NP-complete.
\end{theorem}

\begin{figure}[htpb]
\centering
\begin{minipage}[b]{0.45\linewidth}
  \centering
  \includegraphics[scale=0.4]{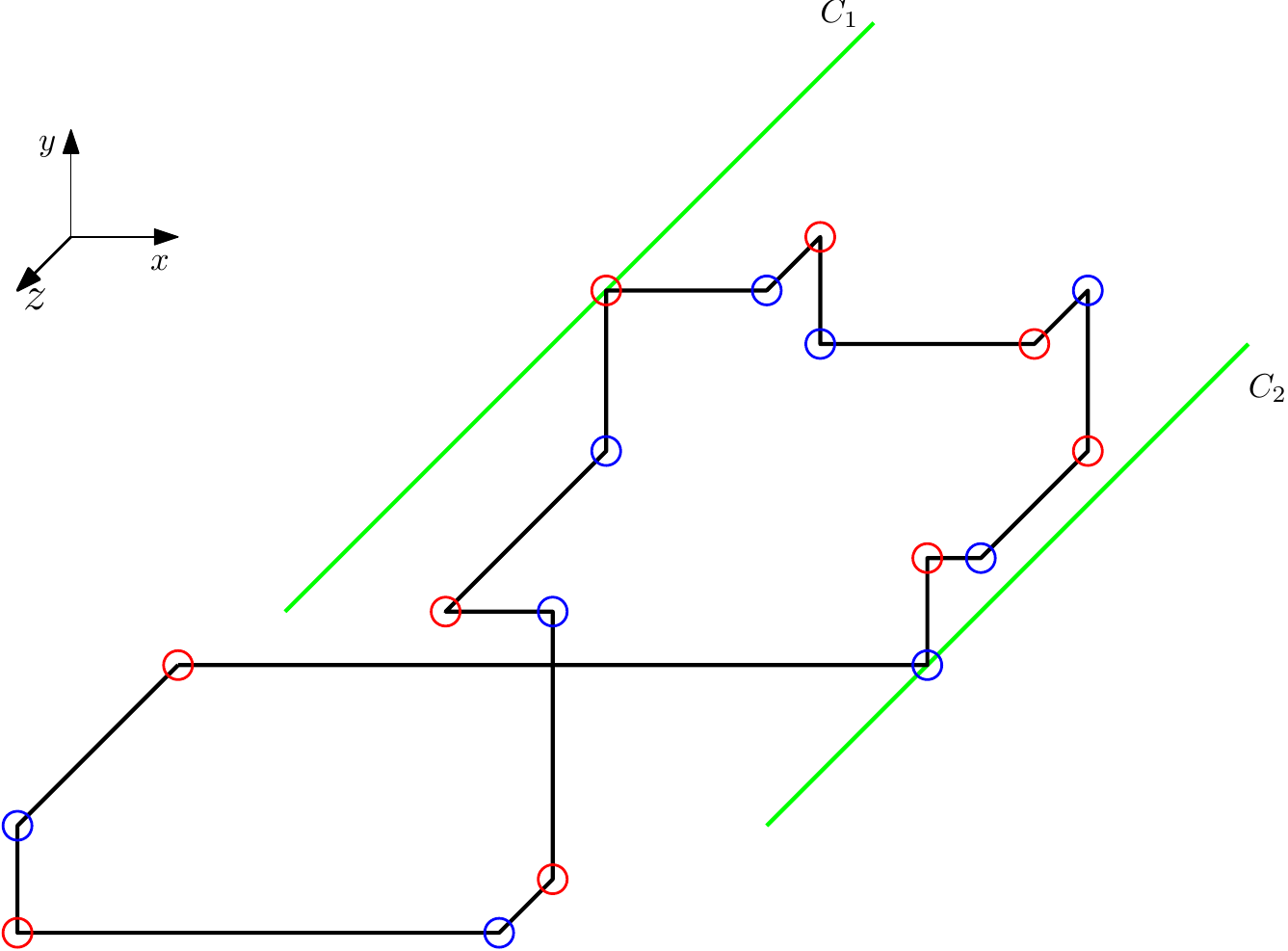}
\caption{An example of a variable loop.}
\label{fig:3d-loop}
\end{minipage}
\quad
\begin{minipage}[b]{0.45\linewidth}
  \centering
  \includegraphics[scale=0.5]{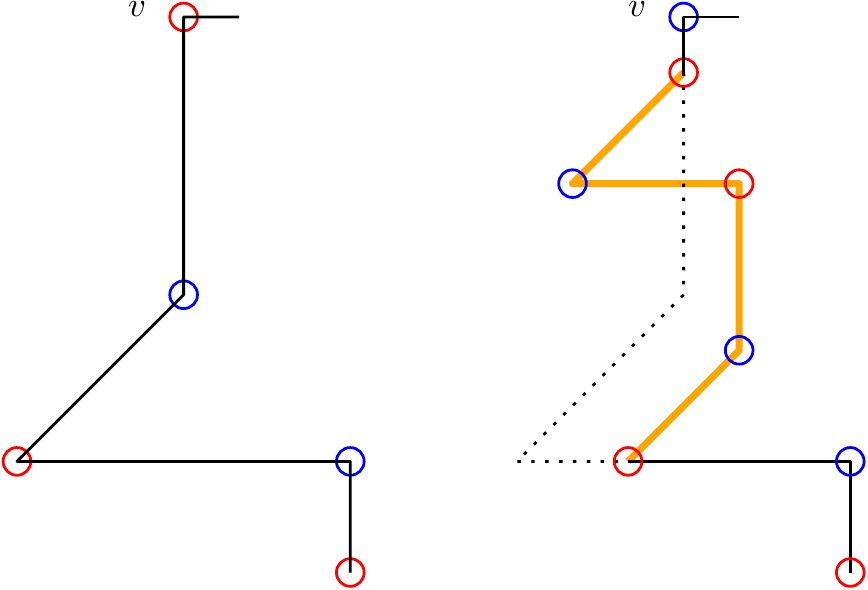}
\caption{The insertion of the orange detour changes the color of $v$ from red to blue.}
\label{fig:adjust-loop}
\end{minipage}

\end{figure}

\begin{proof}
We give a reduction from 3SAT. We say that a line is a $d$-line if it
is parallel to the $d$-axis; we say that a plane is an $ab$-plane if
it is parallel to the plane spanned by the $a$-axis and the
$b$-axis. A clause is represented by a $z$-line. A variable is
represented by a loop of axis-parallel lines with the following
properties:
\begin{enumerate}
\item No two $x$-lines lie in the same $xz$-plane.
(This ensures that when a clause $z$-line meets a vertex of the loop where an $x$-line meets a $y$-line, it does not also meet another such vertex.)
\item There are an even number of edges in each loop.
\item Lines from two different loops do not intersect.
\item A loop intersects a clause $z$-line if and only if the variable
represented by that loop is in the clause represented by the
$z$-line. The intersection point represents a literal in the corresponding
clause.
\item There are two optimal hitting sets for a loop -- the set of odd
vertices, and the set of even vertices. All positive (resp., negative)
literals should be in the same hitting set.
\end{enumerate}

Figure~\ref{fig:3d-loop} shows a portion of an instance in which
clause $C_1$ includes literal $x_1$ and clause $C_2$ includes
literal~$\overline{x_1}$.

In three dimensions, it is not hard to take detours to avoid unwanted
intersections. The number of vertices in a loop can be adjusted by
inserting a detour, as shown in
Figure~\ref{fig:adjust-loop}. Finally, we argue that all clause
$z$-lines can be hit for free by the optimal hitting sets of variable
loops if and only if there is a satisfying truth assignment for the
corresponding 3SAT instance.
\qed
\end{proof}

\section{Hitting Rays and Lines}
\label{sec:ray}
Hitting rays is ``harder'' than hitting lines,
since any instance of hitting lines has a corresponding equivalent
instance as a hitting rays problem (place the apices of the rays
far enough away that they are effectively lines).
A ray has a unique line, its {\em containing line}, that is a superset of the ray.
Two rays having the same containing line are {\em collinear}.
While two lines that are collinear are identical, two rays that are collinear fall
into two groups according to the direction they point along the containing line, $\ell$. Because of nesting, we need
keep only one of the rays pointing in each of the two directions along~$\ell$.
For example, among left-pointing rays, we keep only the one contained in all other
left-pointing rays, i.e., the one with the left-most apex.

We show that the special case with horizontal rays and vertical lines
(abbreviated HRVL) is exactly solvable in polynomial time:
\begin{theorem}\label{thm:rays-and-lines}
The hitting set problem for vertical lines and horizontal rays can be
solved in $O(nT)$ time, where $n$ is the number of entities and $T$ is
the time for computing a maximum matching in a bipartite graph with $n$ nodes.
\end{theorem}

 We begin with a high-level overview of the algorithm. A point can
cover at most 3 objects: a vertical line, a left-facing ray, and a
right-facing ray.  This requires the two rays to intersect in a
segment, and the vertical line to intersect this segment. We call
points at such intersections {\em 3-hitters}.  We can compute the
maximum possible number of 3-hitters, with no two sharing a line or a
segment, via maximum matching in a bipartite graph, where edges
represent intersections between vertical lines and horizontal
segments. We prove there exists an optimal solution with this maximum
number of 3-hitters.  The algorithm performs a sweep inward from the
left and right, finding a suitable set of 3-hitters, ensuring the
remaining lines have the best possible chance to share a point with
the remaining rays.  Once everything that is 3-hit is removed, the
remaining objects intersect in at most pairs.  So we can finish the
hitting by solving an edge cover problem.   
The {\em edge cover} problem for a graph $G=(V,E)$ is to compute a
minimum-cardinality set $E^*\subseteq E$ of edges such that every
vertex in $V$ is incident to at least one edge of $E^*$; the problem
is solved in polynomial time, using maximum cardinality matching,
followed by a greedy algorithm.  Our edge cover instance is a graph with a vertex
for each object and an edge between each pair of intersecting objects.

We now give additional algorithmic and proof details.
We call a horizontal ray to the left (resp., right) an {\em l-ray}
(resp., {\em r-ray}).  In this section, all lines are vertical.  If two
collinear rays are disjoint, we shift one ray slightly up or down, so no two
disjoint rays are collinear.  These rays cannot be covered by a single point, so 
this does not fundamentally alter the optimal solution.  We also assume
that no collinear rays have the same orientation, since in this case one ray is
a subset of the others, which are redundant. 

If a ray is not collinear with any other ray, we add a ray to pair with it. For
example, if an r-ray intersects no l-ray, we add an intersecting l-ray
whose right endpoint is to the right of all vertical input
lines. This additional ray will not change the optimal solution. If an l-ray
and r-ray intersect, their intersection is a segment.  Since each ray intersects
exactly one other ray, we represent each such pair of rays by their segment. 

Let $H$ denote the set of segments and $V$ denote the set of lines, and let $h$ and $v$ denote their cardinalities
respectively. A naive feasible solution is to use $v$ points to cover the lines and
$h$ points to cover the segments independently.   
The only way to improve upon the naive solution is to find points that hit both a line and one or two rays.
The points that hit the lines can help ``hit'' segments in two possible ways:

\begin{enumerate}
\item[(1)] The point on a line may be placed on a segment. We call the
    corresponding line a {\em 3-hitter} and say that the segment is {\em 3-hit} by the
    line.
\item[(2)] Points on lines may hit each ray outside its intersecting segment. 
    This requires two points on two distinct lines. We call the left (resp., right) line an {\em l-hitter} (resp., {\em r-hitter}). We say
    the segment is {\em double-hit} by those two lines.
\end{enumerate}

These are the only ways to improve over the naive hitting set.  To see this, suppose a vertical line is an r-hitter, that is, shares a point with a right-pointing ray
outside the shared segment with its left-facing ray. Suppose no vertical line shares a point with the corresponding left-facing ray.  Then that ray requires a separate
point.  This is 
equivalent to putting a point on the segment and hitting the vertical line separately (two points to hit one segment and one line).

When tallying these improvements for any feasible solution, we allow at most one point on a vertical line to be involved in any 3-hitter or double-hitter.  This is because once a point on a vertical line is selected, the line is covered, and additional points no longer help cover the line.
More precisely, we say a set of $v_1$ 3-hit segments and $v_2$ double-hit segments are {\em independent} if the union of the relevant lines (3-hitters, r-hitters, and l-hitters) has cardinality
$v_1 + 2v_2$.
That is, no vertical line is involved in 3-hitting or double-hitting more than one segment of an independent set.

Consider an instance $I = H \cup V$ and a feasible solution $S$.  
Let $T$ be some maximal independent set of 3-hit and double-hit segments with respect to $S$.  Suppose there are $v_1$ 3-hit segments and $v_2$ double-hit segments in $T$.  Then, $|S| \geq h+v-v_1-v_2$, and there is a feasible solution with precisely $h+v-v_1-v_2$ points.  To see this, first remove from $I$ the segments in $T$ and their corresponding 3-hitters and double-hitters. Let $I'$ refer to the resulting instance, whose size is $|I| - 2v_1 - 3v_2$ due to the independence of $T$.  Likewise let $S'$ refer to points in $S$ that intersect an object in $I'$; since $S$ is a feasible solution for $I$, $S'$ is a feasible solution for $I'$.  The instance $I'$ cannot contain any 3-hit or double-hit segments, since such a segment would be independent from those in $T$, contradicting the maximality of $T$.  We observed above that (1) and (2) are the only ways to improve upon a naive hitting set.  Hence a naive solution is optimal for $I'$, and $|S'| \geq |I'|$.  We have that $|S\setminus S'| = v_1 + 2v_2$, yielding $|S| \geq |I'| + v_1 + 2v_2 = h+v-v_1-v_2$.  The inequality becomes tight if we replace $S'$ with a naive solution for $I'$.  In particular if $S$ is an optimal solution for $I$, then the inequality must be tight.  Thus, we can think of an optimal solution as maximizing $v_1 + v_2$, and our goal is to maximize the number of independent 3-hit and double-hit segments.

Given an instance of HRVL,  we can
calculate the maximum number of 3-hitters. We construct the bipartite
graph $G$ in which one set of nodes is the lines and the other set of nodes
is the segments; there is an edge between two nodes if and only if the line and
the segment they represent intersect. We refer to $G$ as the {\em lines-segments graph}.  Maximum matching 
is solvable in polynomial time. A matching in the graph
represents a set of independent intersections in the corresponding HRVL. 
That is, a set of $M$ edges in a matching corresponds to a way to hit $M$
segments and $M$ lines with $M$ points.  These are hittings of type (1).
The following lemma shows that hitting points of type (1) are preferred over hitting
points of type (2).

\begin{lemma}
For any instance of HRVL, there is a maximum matching between lines and segments that can be
augmented to be an optimal solution.
  \label{lem:hrvl_oo}
\end{lemma}

\begin{proof}
The proof is by contradiction. Let $v_1^*$ be the largest $v_1$ for any
minimum hitting set. We assume that $v_1^*$ is less than
$m$, the cardinality of the maximum matching between lines
and segments. Thus, there is an augmenting path in the bipartite graph $G$;
an example of such a path is shown in green in Figure~\ref{fig:aug_path}.
Because the current solution is optimal,
any augmenting path cannot improve it. This allows
us to infer some properties of the first segment and the last line on the augmenting
path.  We consider the augmenting path $P$ with the shortest length (fewest elements) and the shortest
horizontal distance between the last two lines along the path.
Then by case analysis on path $P$, we argue there exists another augmenting
path that increases $v_1^*$ or violates a  minimality condition of $P$.

In more detail, an augmenting path in graph $G$ corresponds to a sequence of alternating
segments and lines in the HRVL instance: $\left\{ e_1, l_1, e_2, l_2, \dots, e_n, l_n \right\}$. In the current
solution line $l_{i-1}$ is matched with segment $e_i$.
In Figure~\ref{fig:aug_path} the green
path is an example of an augmenting path, where $n=3$ and blue dots correspond to current 3-hitter points. Because of the optimality of
the current solution, any augmenting path cannot improve it. 
The following two properties hold;
otherwise, after augmenting, the sum of $v_1$ and $v_2$ will stay the
same, but $v_1^*$ would be increased by 1:
\begin{itemize}
  \item $e_1$ is double-hit by other lines.
  \item $l_n$ is helping to double-hit another segment.
\end{itemize}

Without loss of generality, we assume the intersection of $e_1$ and $l_1$ is to the left of $l_n$.
Also assume that $n$, the number of lines in the augmenting path, is as small as possible.
Among augmenting paths with smallest $n$, we pick one with the shortest horizontal distance between 
lines $l_{n-1}$ and $l_n$.
We consider the following cases:
\begin{enumerate}
  \item If line $l_n$ is the l-hitter of some segment $e_t$, then the l-hitter of segment $e_1$
    can take its job. One can do the augmenting and assign the
    l-hitter of $e_1$ to l-hit $e_t$. Therefore the solution is
    still optimal and $v_1^*$ increases.
  \item Suppose line $l_n$ is an r-hitter for a segment $e_t$, and let $l_r$ be the r-hitter for segment $e_1$.  If line $l_r$ is to the
    right of line $l_n$, then line $l_r$ can take the job of line $l_n$.  That is, line $l_r$ can be the r-hitter for segment $e_t$.
    Again, it is now possible to increase $v_1^*$ while maintaining an optimal total number of points.
  \item If line $l_n$ is an r-hitter and the r-hitter of $e_1$ (called
    $l_r$) is to the left of $l_n$, we know that line $l_r$ will
    intersect a segment, $e_q$ (with $1\leq q\leq n$) of the augmenting path. This is because we assume line $l_n$ is to the right of
    the intersection of $e_1$ and $l_1$. In Figure~\ref{fig:aug_path}, line $l_4$ is a possible $l_r$.
    Thus, using $l_r$ gives
    a shorter augmenting path, $\left\{ e_1, l_1, \dots, e_q, l_r
    \right\}$. This new path either has strictly fewer lines or has a strictly shorter horizontal distance between the last lines.
    This contradicts the choice of the first augmenting path.
\end{enumerate}
\qed
\end{proof}

\begin{figure}[htpb]
\centering
\begin{minipage}[bt]{0.45\linewidth}
  \centering
  \includegraphics[scale=0.75]{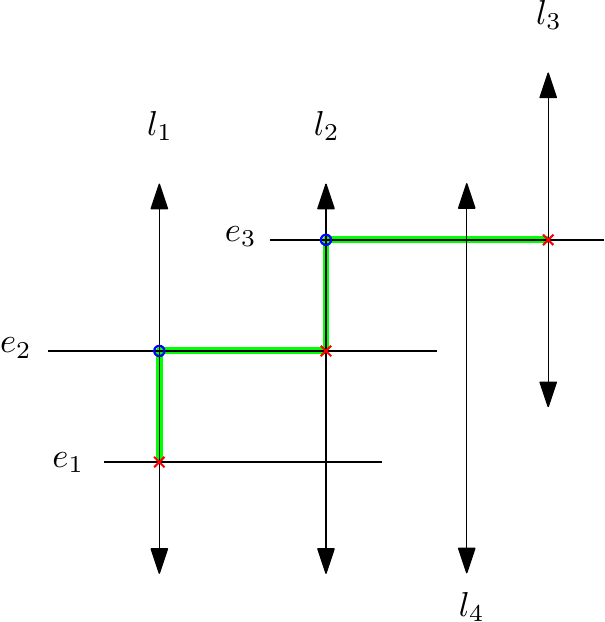}
\caption{A green augmenting path: the matching size increases by replacing blue circles with red crosses.}
\label{fig:aug_path}
\end{minipage}
\quad
\quad
\begin{minipage}[bt]{0.45\linewidth}
\includegraphics[width=2.2in]{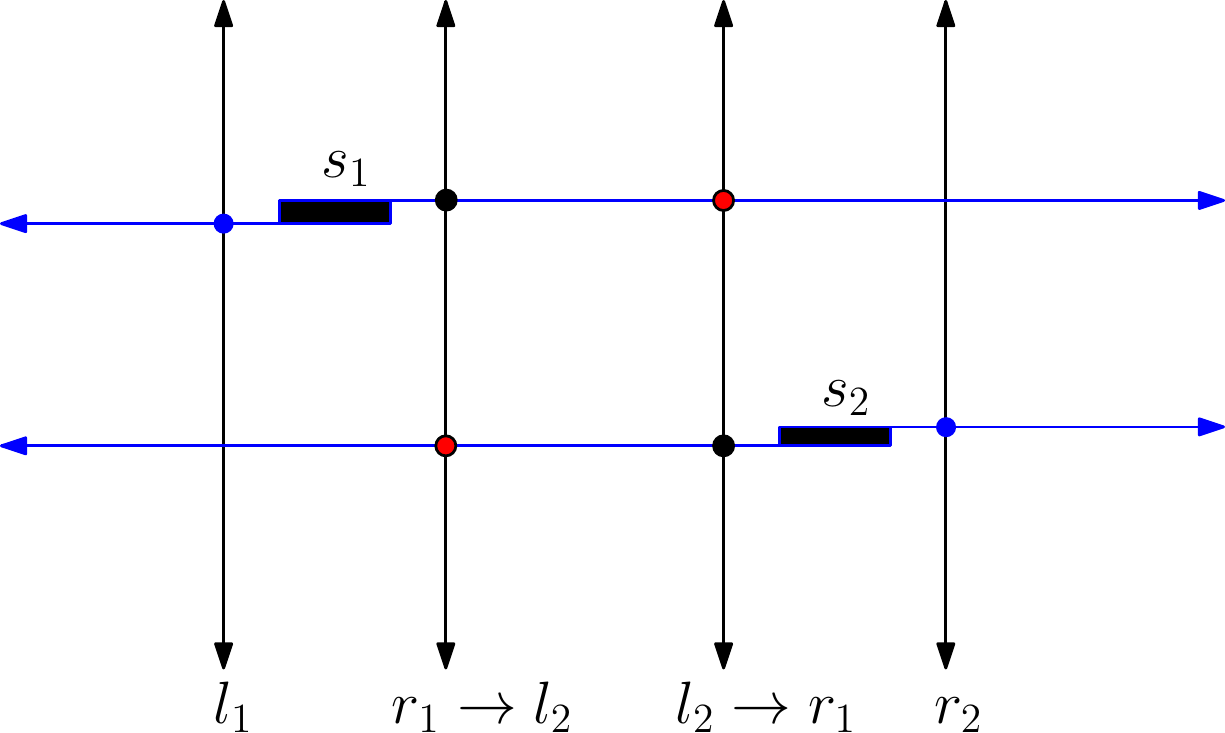}
\caption{Segments $s_1$ and $s_2$ are originally double hit using the blue and black points on their rays.
They can also be double hit using the blue 
 and red points on their rays.  The middle two lines can
swap their roles because it is always possible to move an r-hitter right or an l-hitter left.  
In general, the more a line is to the left, the more flexible it can be as an l-hitter and the more
a line is to the right, the more flexible it is as an r-hitter.}
\label{fig:one_side}
\end{minipage}
\end{figure}

The following lemma gives additional useful structure for at least one optimal solution:

\begin{lemma}
Given an optimal solution $\mathcal{S}$, there is an optimal solution
$\mathcal{S}'$ that has the same set of 3-hitters as $\mathcal{S}$, with
its l-hitters all left of its r-hitters.
  \label{lem:one side}
\end{lemma}

\begin{proof}
Let $d$ be the number of double-hit segments in an optimal solution.
We will show that for the $2d$ lines involved in double-hitting these segments, there exists a  solution $\mathcal{S}'$ where the first $d$, numbering from left to right, are
$l$-hitters.  Thus, the next (last) $d$ are r-hitters.

The proof is by contradiction.  Figure~\ref{fig:one_side} illustrates the following argument. Let $\mathcal{S}$ be an optimal solution with
the rightmost first r-hitter, $r_1$.  Assume, that there are strictly fewer than $d$ l-hitters to its left.
Thus, there is at least one l-litter to the right of $r_1$. Let $l_2$ be any such l-hitter.
Let $s_1$ be the segment r-hit by $r_1$.  Its l-hitter $l_1$ must be to the left of $r_1$, since for any given segment, its l-hitter is to the right
of the segment and its r-hitter is to the right.
Let $s_2$ be the segment l-hit by $l_2$.  Its r-hitter $r_2$ must be to the right of $l_2$.  In figure~\ref{fig:one_side}, for each segment, the black and blue
points on their segments intersect the l-hitters and r-hitters.  We can swap the roles of $r_1$ and $l_2$ while still double hitting both segments.  Instead of
using the blue points in ~\ref{fig:one_side}, we use the red points.  The black and red points still hit all four rays associated with segments $s_1$ and $s_2$.  However,
now the rightmost first r-hitter in the new solution has moved further right.  This contradicts our choice of solution $S$.
\qed
\end{proof}

Algorithm~\ref{algo:hrvl} below gives an optimal solution for HRVL.
The algorithm maximizes the number of 3-intersections and ``balances''
the remaining lines between the left and right sides as much as
possible. We test the ``criticality'' of a line $l$ by computing a
maximum cardinality matching in the lines-segments graph, with and
without the line $l$; if the matching cardinality drops when line $l$
is not part of the graph, then line $l$ is {\em critical}. In the
algorithm, we check the criticality of lines: given the previous
choices, if a critical line is not used as a 3-hitter, there is no way
to extend the previous choices to a maximum matching.  We add a
3-hitter to our solution if and only if the line involved is critical.
A 3-hitter is matched to the segment crossing it that ends first in
the current sweeping direction.  We argue below that this algorithm finds
a maximum set of 3-hitters.
Non-critical lines are counted as l-hitters when sweeping from the
left and as r-hitters when sweeping from the right.  The algorithm
removes each newly-discovered non-critical line from consideration as
a 3-hitter, and swaps the sweep direction (from left to right, or from
right to left).  This balances the number of presumed l-hitters and
r-hitters during the course of the the algorithm. When the sweep
encounters a segment $s$, this triggers testing of the first line (in
the sweep direction) that intersects $s$, if any.  Subsequent sweep
steps continue to process segment $s$ until it is matched as part of a
3-hitter, or it is removed, to be double hit at the end of the
algorithm.
A small illustrative example is shown in Figure~\ref{fig:algo1-example}.

\begin{algorithm}
  Input: set $L$ of vertical lines, set $S$ of horizontal segments (ray intersections)\;
  $H$ $\gets$ [0, 0] \hspace{7px} //$H$ counts 2-hitters at left and right sides\;
  $I_3 \gets \left\{  \right\}$ \hspace{7px}  //$I_3$ stores 3-intersections of the solution\;
  SD$\gets$0; \hspace{3px}  
  //SD stands for sweep direction. 0 is from left to right; 1 is reverse\; 
  $L_2 \gets \emptyset$; $S_2 \gets \emptyset$ // unmatched lines and segments

{\bf step A}: \eIf{there are any 3-intersection left} {
  sweep along the direction indicated by SD\;
  When a line and segment start at the same time, the segment is seen first.\\
  \eIf{the event is a line $l_1$}{
    // Do not consider the line as a part of a 3-hitter.\;
    // Send to the final edge-cover problem\;
    {\bf step B}: \;
    $L \gets L - \{l_1\}$ \;
    $L_2 \gets L_2 \cup \{l_1\}$ \;
    $H$[SD]++\;
    toggle SD;
  }{
    // the event is a segment $e_1$\;
    \eIf{$e_1$ crosses some line(s)}{
       $l \gets$ the line hitting $e_1$ that is closest along SD\;
    }
    {
      // Do not consider $e_1$ as a part of a 3-hitter.\;
      // Send to the final edge-cover problem\;
      $S \gets S - \{e_1\}$ \;
      $S_2 \gets S_2 \cup \{e_1\}$ \;
      go to step A\; 
      }
    \eIf{$l$ is critical}{
      //For example, if SD is 0, look at the right endpoints of segments \\
      //crossed by $l$. Pick the one with the leftmost right endpoint\\
      $e_2 \gets$ the segment crossing $l$ with the closest endpoint to $l$ along SD\;
      // put the intersection of $e_2$ and $l$ into $I_3$ \;
      $I_3 \gets I_3 \cup (l, e_2)$ \label{line:make-3-hitter} \;
      $S \gets S - \{e_2\}$ \label{segment-matched}\;
      $L \gets L - \{l\}$ \;
      go to step A\;
    }{
      go to step B\;
    }
  }
}{
  Solve the remaining problem $L \cup L_2$ and $S \cup S_2$ optimally using edge cover problem\;
}
\caption{Bidirectional sweeping algorithm for HRVL.}
\label{algo:hrvl}
\end{algorithm}

\begin{figure}[htpb]
  \begin{center}
    \includegraphics[width=0.5\textwidth]{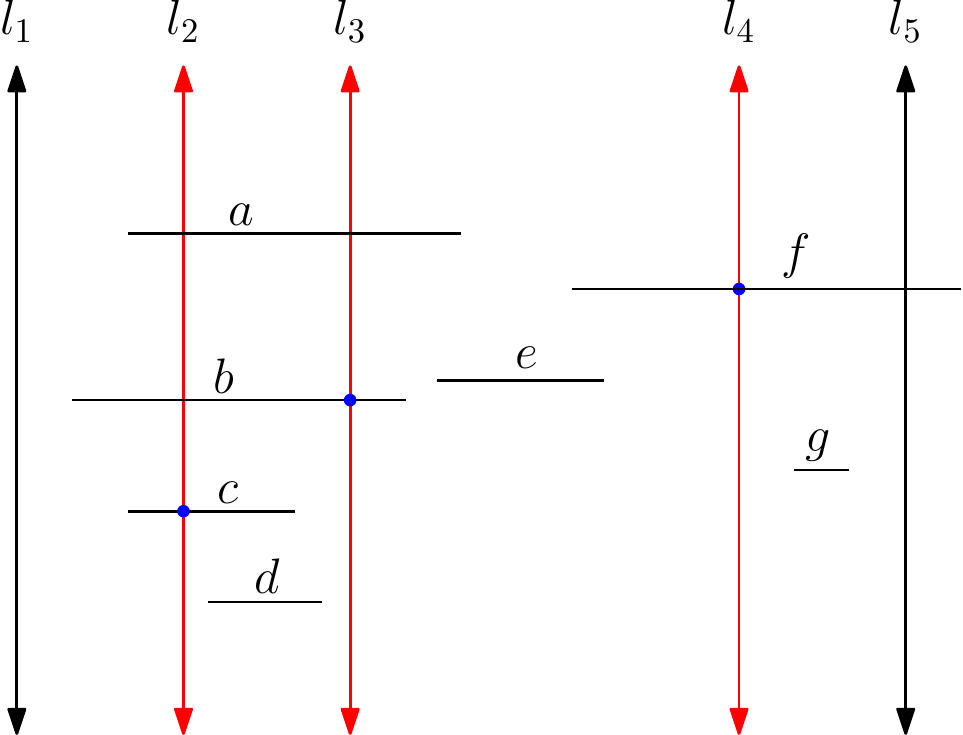}
  \end{center}
  \caption{Example of the bidirectional sweep Algorithm 1: We begin sweeping from the left. Line $l_1$ does not intersect a segment, so cannot be a 3-hitter.  We put it in set $L_2$ and switch to sweeping from the right.  We find segment $f$, which leads to an examination of line $l_5$.  Line $l_5$ is not critical, so we put it in $L_2$ and switch to a left sweep. We find segment $b$ which leads to an examination of line $l_2$. Line $l_2$ is critical.  It intersects segments $a$, $b$ and $c$. Since the right endpoint of segment $c$ is leftmost, we select the intersection of $c$ and $l_2$. We continue sweeping from the left, find segment $b$ again which leads to line $l_3$. Line $l_3$ is critical and is matched to segment $b$, the one that ends soonest in the left sweeping direction. Continuing a sweep from the left, we find segment $a$. Segment $a$ no longer crosses any lines, so it is added to set $S_2$.  Next we find segment $d$ and $e$ in order, both added to set $S_2$. Finally we find segment $f$, which leads to consideration of line $l_4$, which is critical.  We match line $l_4$ to segment $f$, since this is the first segment to end going left (and the only remaining segment intersecting $l_4$). At this point, we know there are no more 3-hitters, so the remainder of the problem is added to the sets $L_2$ and $S_2$. In this case, segment $g$ is added to set $S_2$. 
We then solve the remainder of the hitting set problem (lines $l_1$ and $l_5$ and segments $a$, $d$, $e$ and $g$) optimally as an edge cover problem. Any one of the segments is double hit by the pair of lines.  The remaining three segments are hit with a point each.}
  \label{fig:algo1-example}
\end{figure}

We argue the correctness of Algorithm~\ref{algo:hrvl}, beginning with the following lemma.

\begin{lemma}
Algorithm~\ref{algo:hrvl} selects a maximum-cardinality set of $3$-hitters.
  \label{lem:hrvl_3hit}
\end{lemma}

\begin{proof}
When the algorithm chooses not to match a line $l$ to a segment as a
3-hitter, it has determined that this choice is correct:  line $l$
can be omitted from the set of 3-hitters because the remaining lines
and segments still have a maximum cardinality matching.  What remains
to be shown is that when the algorithm creates a 3-hitter in
Line~\ref{line:make-3-hitter}, the set of 3-hitters chosen so far plus a maximum
cardinality set of 3-hitters in the remaining problem is a maximum
cardinality set for the original problem.

Consider an instance of HRVL.  Assume, without loss of generality, that the first 3-hitter the
algorithm finds (line $l$) is while sweeping from the left.  Let $I$ be the
active instance when line $l$ is considered.  All lines to the left of
line $l$ have been added to set $L_2$, i.e., they have been
designated as possible l-hitters. Let $M$ be a maximum-cardinality
matching in the line-segment graph for instance $I$.  Because line $l$ is
critical, it is matched to some segment $s_m$ in the matching. Let $s$ be the
segment 3-hit by line $l$ (matched to line $l$) according to Algorithm~\ref{algo:hrvl}.
If $s = s_m$ then we are done. The $3$-hitter is part of $M$ and therefore is part
of a maximum-cardinality matching.  Otherwise, if segment $s$ is not part of matching $M$, then
matching $l$ to $s$ instead of $s_m$ gives a matching $M'$ of the same cardinality
as $M$. Finally, suppose that segment $s$ is matched to line $l_m \neq l$ in matching $M$.
For example, consider the instance in Figure~\ref{fig:algo1-example} ignoring all objects
except lines $l_2$ and $l_3$ and segments $a$ and $b$.  In this example, $l$ would be $l_2$,
$s$ would be $b$, $l_m$ would be $l_3$ and $s_m$ would be $a$.
Because there are no lines to the left of line $l$ in instance $I$, line $l_m$ is to the
right of line $l$. By the choice of segment $s$ as the segment intersecting $l$ with the leftmost
right endpoint, we have that segment $s_m$ extends at least as far to the right as $s$ does.  Because line $l_m$
intersects segment $s$, and segment $s_m$ extends at least as far to the right as segment $s$ does, then we know
that line $l_m$ and segment $s_m$ intersect.  Therefore, if we match line $l$ with segment $s$, in the remaining
problem, line $l_m$ can be matched with segment $s_m$.  Combined with the rest of matching $M$, this is a maximum
cardinality matching.

This completes the argument for the first 3-hitter.  The same argument holds for the rest of the 3-hitters, whether scanned
from the left or from the right.  All vertical lines ``behind'' the new 3-hitter in the scan direction have either been matched to
segments or have been removed from the problem as potential l-hitters or r-hitters.  Thus, the remaining problem at the time the new line
is considered contains lines on only one side (later in the scan direction).  This matches the conditions used above.
\qed
\end{proof}

We now argue that the left-right-balanced approach in Algorithm~\ref{algo:hrvl} leaves lines that
are excellent candidates as l-hitters and r-hitters.
Let $S$ be the solution given by
Algorithm~\ref{algo:hrvl}, and let $S'$ be an optimal solution with
the maximum cardinality set of 3-hitters. We know that $S$ and $S'$ have the same
number of 3-hitters. Let $D$ and $D'$ denote the sets of lines left behind
(not 3-hitters) in $S$ and $S'$, respectively. We order lines in $D$
and $D'$ from left to right. Let $k$ be $\floor{\frac{|D|}{2}}$. Thus,
there are at most $k$ pairs of double-hitters in $S$ and $S'$. Let
$lh_i$ (resp.,~$lh_i'$) be the $i$th line of $D$ (resp.,~$D'$).

Given a solution $P$ and a line $l$, let $E(l, P)$ denote the number
of segments on the left side of $l$ not hit by 3-hitters in $P$. A
line having more segments on its right side 
is more likely to be an
l-hitter.  We will show that line $lh_i$ is at least as capable of
being an l-hitter as is line~$lh_i'$; specifically, we show that
\begin{align}
  E(lh_i, S) \leq E(lh_i', S'),  \quad i=1,2,..,k.
\label{eqn:l-hitter}
\end{align}

Before proving inequality (\ref{eqn:l-hitter}), we argue that this
inequality, the equivalent inequality with respect to r-hitters and
previous arguments suffice to prove the correctness of
Algorithm~\ref{algo:hrvl}.  This also proves
Theorem~\ref{thm:rays-and-lines}. We argued that the optimal solution
maximizes the number of 3-hitters plus the number of double-hit
segments.  Lemma~\ref{lem:hrvl_oo} shows that it suffices to first
maximize the number of 3-hitters and then, subject to that constraint,
maximize the number of double-hit segments.  Lemma~\ref{lem:hrvl_3hit}
shows that Algorithm~\ref{algo:hrvl} first maximizes the number of
3-hitters.  Consider an optimal solution $S'$ with the maximum number
of 3-hitters.  Also, using Lemma~\ref{lem:one side}, assume that if
there are $d'$ segments double hit, then they are hit with the
leftmost remaining $d'$ lines and the rightmost remaining $d'$ lines.
We now argue that the final edge-cover computation in
Algorithm~\ref{algo:hrvl} finds as many double-hittings as solutions
$S'$ has.

Let $I'$ be the HRVL instance with all the solution $S'$ 3-hitters and
the segments they 3-hit removed.  Let I be the corresponding instance
after removing the lines and segments involved in 3-hitters in
solution $S$.  The sets of lines and segments left behind can be
different in the two instances. No lines cross any segments in either instance, since
otherwise the set of 3-hitters would not be maximum. Consider a segment $s$ in either problem.
It can be l-hit by any line to its left and it can be r-hit by any line to its right.
Thus, if there are $q$ lines to its left and $r$ lines to its right, there are $qr$ possible
ways it can be double hit.

As above, let $k$ be the maximum number of double-hitters (the floor of half the number of lines).
Inequality~(\ref{eqn:l-hitter}) says that, numbering from the left,  the $i$th line in instance $I$ has more segments
to its right than the $i$th line in instance $I'$ has for all $k$ of the leftmost lines. Thus, each of the first
$k$ lines in instance $I$ can l-hit at least as many segments as their counterparts in instance $I'$.
An argument similar to the proof of Inequality (\ref{eqn:l-hitter}) below shows that each of the rightmost $k$ lines
in instance $I$ can r-hit at least as many segments as its counterpart in instance $I'$.

Consider a double-hit segment in solution $S'$ for instance $I'$.  We
can represent the double-hitting as $(x,y,z)$ where $x$ is the {\em
index} of the l-hitter in the set of l-hitters (a number between $1$
and $k$), $y$ is the index of the segment numbered from the left, and
$z$ is the index of the r-hitter, numbered from the right (a number
between $1$ and $k$).  Let $T'$ be the set of all such triples
representing the double hitting in solution $S'$.  Then the same set
of triples is a feasible double-hitting for instance $I$.  The lines
and segments may be different, but the indices within the instances
are the same.  This is feasible because, in this index-based representation,
the set of feasible indices for the l-hitter for the $i$th segment in
instance $I$ is a superset of the set of feasible indices for the
$i$th segment in instance $I'$. Similarly the set of feasible indices for the
r-hitter in instance $I$ is a superset of the set of feasible r-hitters in
instance $I'$. 

Since the index-based solution for $S'$ is feasible in $S$, the final edge cover 
solution will give at least as many double hit segments in $S$ as there are in $S'$.

We now prove inequality (\ref{eqn:l-hitter}). We split the proof into one claim and two lemmas.

Because of the criticality test and the choice of intersecting segment
$e_2$, we have the following claim:

\begin{claim}
In $S$ if a 3-hitter is on the left side of an l-hitter, the segment hit by the 3-hitter will not
intersect that l-hitter.
\label{claim:no-intersect}
\end{claim}

\begin{proof}
The proof is by contradiction.  Let $l$ be a 3-hitter, matched to
segment $e$.  Let $l_h$ be the closest l-hitter to the right of line
$l$ and suppose $l_h$ intersects segment $e$.  There can be $j \ge 0$
lines between $l$ and $l_h$ which must all be 3-hitters.  Let this set
of 3-hitters with their matched segments be $(l_1, e_1), (l_2,
e_2), \ldots, (l_j, e_j)$.  Let $t$ be the size of the maximum
matching in $G$ at the time line $l$ is tested for criticality.  That
is, the maximum matching size drops to $t-1$ when line $l$ is removed.
When line $l_h$ is tested for criticality, the size of the maximum
matching is $t-j-1$ whether $l_h$ is included or not.  Let $M$ be a
maximum matching when $l_h$ is not included.  We can augment $M$ to a
matching of size $t$ that does not include line $l$: add the $j$
pairings from the 3-hitters between $l$ and $l_h$ and then add $(l_h,
e)$.  Segment $e$ is not part of matching $M$ since $e$ is removed
from the set of active segments in line~\ref{segment-matched} of
Algorithm~\ref{algo:hrvl}.  This contradicts the criticality of line~$l$.
\qed
\end{proof}

\begin{lemma}
  $lh_i$ cannot be on the right side of $lh_i'$,  $i=1,2,..,k$.
  \label{lem:left}
\end{lemma}

\begin{proof}
The proof is by induction.
We prove the base case by contradiction. Figure~\ref{fig:trace}
illustrates the following argument.  When $i=1$, we assume that $lh_1$
is on the right side of $lh_1'$.  This means that in $S$, line $lh_1'$
is a 3-hitter. Suppose $e_1$ is the corresponding segment hit by
$lh_1'$ in $S$. We know in $S'$ line $lh_1'$ does not hit $e_1$ (since
line $lh_1'$ is not a 3-hitter in $S'$). Thus, $e_1$ must be hit by a
different 3-hitter in $S'$, say $l_3$ (otherwise, $lh'$ can 3-hit
$e_1$ to increase the number of 3-hitters, contradicting the choice of
$S'$). Again in $S$, $l_3$ hits $e_2$, which means in $S'$, $e_2$ must
be hit by another line $l_4$. Claim~\ref{claim:no-intersect}
guarantees that all of the lines $l_i$ and the segments $e_j$ involved in the
tracing process are on the left side of $lh_1$. This tracing will stop
eventually, because there are only a finite number of lines to the
left of line $lh_1$. This gives a contradiction.

Now we assume that $i$ is the smallest integer such that $lh_i$ is to
the right of $lh_i'$.  We again start tracing from $lh_i'$.  The
tracing process can only end at $lh_j$ ($j<i$); otherwise, a
contradiction exists, as in the base case. Let the tracing sequence be
${lh_i', e_1, l_3, e_2, l_4, \dots, lh_j}$. Since $lh_j'$ is on
the left side of $lh_i'$, so is $lh_j$. In $S'$, we replace 3-hitters
of $S'$ in the tracing sequence by 3-hitters of $S$ in the
sequence. Now, in $S'$, $lh_j$ becomes an l-hitter instead of
$lh_i'$. This modified $S'$ has improved with the replacement of an
l-hitter with a ``more capable'' (left) l-hitter, while keeping the
number of 3-hitters the same. Now, we start a new trace with the new $S'$.

In summary, if the tracing ends with an l-hitter in $S$, we improve
$S'$ and resume tracing.  This process must end in a contradiction because
there are only a finite number of lines, and each revised $S'$ moves an l-hitter
strictly left.
\qed
\end{proof}

\begin{figure}[htpb]
  \begin{center}
    \includegraphics[width=0.9\textwidth]{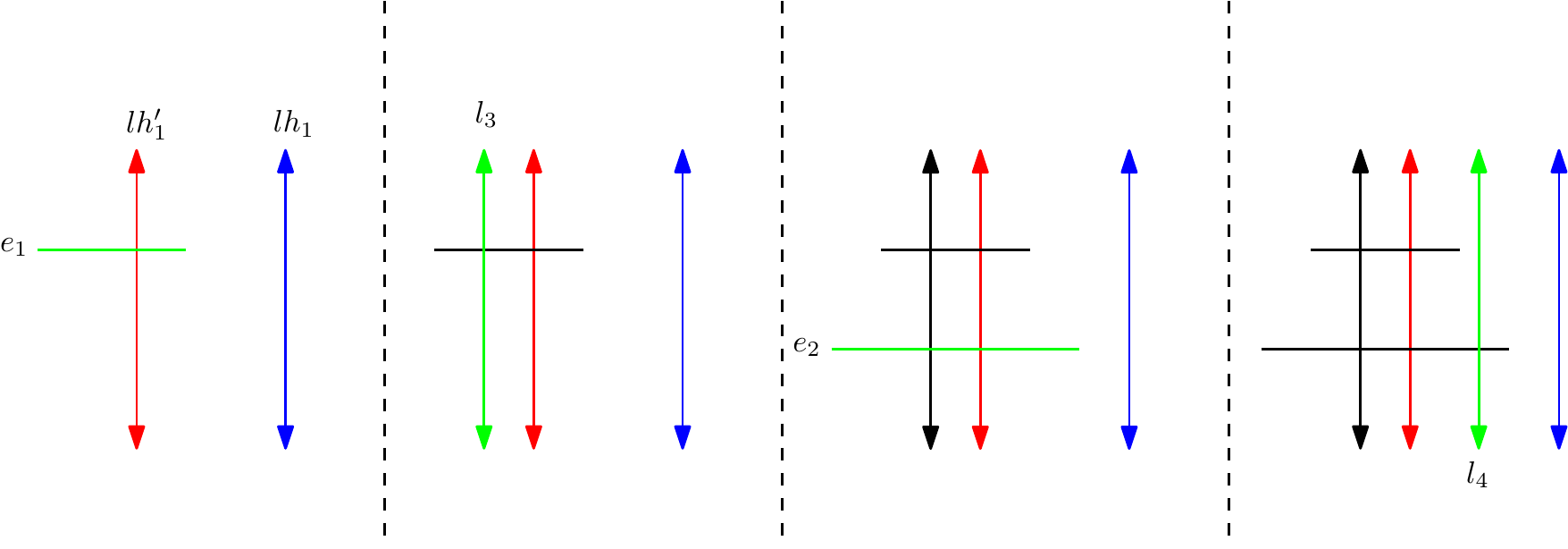}
  \end{center}
  \caption{The tracing sequence for the base case of the proof of Lemma~\ref{lem:left}.  The order items are visited in the tracing is $lh_1$, $lh_1'$, $e_1$, $l_3$, $e_2$, $l_4$. The original l-hitter $lh_1$ is always blue and the original l-hitter $lh_1'$ is always red.  Then for each step of the tracing, the new line or segment is green, with all older elements in black.}
  \label{fig:trace}
\end{figure}

An immediate result from this lemma is 
  \begin{align}
    E(lh_i, S') \leq E(lh_i', S').
    \label{}
  \end{align}

Given a solution $P$ and a line $l$, let $C(l, P)$ denote the number
of segments on the left side of $l$ that have been 3-hit in $P$.  Let
$N(l)$ be the total number of segments on the left of line $l$. The
following lemma shows that the segments that $S$ leaves to be double hit
are the segments that are easier to double-hit.

\begin{lemma}
  $C(lh_i, S) \geq C(lh_i, S')$, $i=1,2,..,k$.
  \label{lem:max hit}
\end{lemma}

\begin{proof}
We showed in Claim~\ref{claim:no-intersect} that if a segment is 3-hit to the left of an l-hitter $l$, then the segment ends before reaching line $l$.
If $C(lh_i, S)<C(lh_i, S')$, then we replace the part of $S$ that is on
the left side of $lh_i$ with the corresponding part of $S'$. This gives
us a solution that has more 3-hitters than $S$ has, contradicting 
the assumption that $S$ has the maximum set of 3-hitters.
\qed
\end{proof}

Therefore we obtain
\begin{eqnarray*}
  E(lh_i, S) &=& N(lh_i) - C(lh_i, S) \\
   & \leq & N(lh_i) - C(lh_i, S')=E(lh_i, S')\leq E(lh_i', S').
  \label{}
\end{eqnarray*}.

\section{Hitting Lines and Segments}
\label{sec:mix}
\subsection{Hardness}

\begin{theorem}
\label{th:rays}
Hitting set for horizontal unit segments and vertical lines is NP-complete.
\end{theorem}

\begin{proof}
The reduction is from 3SAT. Consider a 3SAT instance with $n$ variables and $m$ clauses. See Figure~\ref{fig:rays_npc}.

\begin{figure}[ht]
\centering
\begin{picture}(0,0)%
\includegraphics{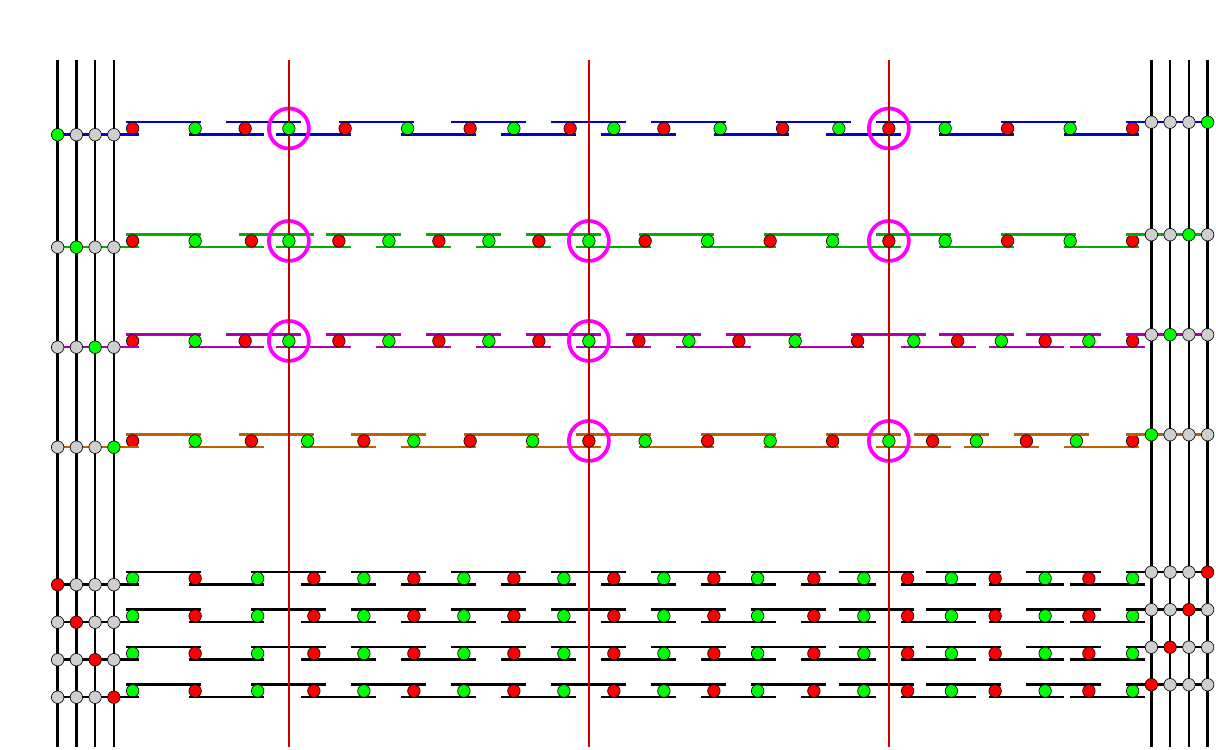}%
\end{picture}%
\setlength{\unitlength}{1579sp}%
\begingroup\makeatletter\ifx\SetFigFont\undefined%
\gdef\SetFigFont#1#2#3#4#5{%
  \reset@font\fontsize{#1}{#2pt}%
  \fontfamily{#3}\fontseries{#4}\fontshape{#5}%
  \selectfont}%
\fi\endgroup%
\begin{picture}(14573,8997)(136,-7894)
\put(151,-4336){\makebox(0,0)[b]{\smash{{\SetFigFont{10}{12.0}{\rmdefault}{\mddefault}{\updefault}{\color[rgb]{.75,.38,0}$x_4$}%
}}}}
\put(151,-1861){\makebox(0,0)[b]{\smash{{\SetFigFont{10}{12.0}{\rmdefault}{\mddefault}{\updefault}{\color[rgb]{0,.69,0}$x_2$}%
}}}}
\put(151,-586){\makebox(0,0)[b]{\smash{{\SetFigFont{10}{12.0}{\rmdefault}{\mddefault}{\updefault}{\color[rgb]{0,0,.82}$x_1$}%
}}}}
\put(3601,764){\makebox(0,0)[b]{\smash{{\SetFigFont{10}{12.0}{\rmdefault}{\mddefault}{\updefault}{\color[rgb]{.82,0,0}$c_1$}%
}}}}
\put(10801,764){\makebox(0,0)[b]{\smash{{\SetFigFont{10}{12.0}{\rmdefault}{\mddefault}{\updefault}{\color[rgb]{.82,0,0}$c_3$}%
}}}}
\put(7201,764){\makebox(0,0)[b]{\smash{{\SetFigFont{10}{12.0}{\rmdefault}{\mddefault}{\updefault}{\color[rgb]{.82,0,0}$c_2$}%
}}}}
\put(151,-3061){\makebox(0,0)[b]{\smash{{\SetFigFont{10}{12.0}{\rmdefault}{\mddefault}{\updefault}{\color[rgb]{.69,0,.69}$x_3$}%
}}}}
\end{picture}%
  \caption{A set of horizontal unit segments and vertical lines that represents
the {\sc 3SAT} instance
$I=(x_1\vee {x_2}\vee {x_3})
\wedge({x_2}\vee x_3\vee\overline{x_4})
\wedge(\overline{x_1}\vee \overline{x_2}\vee x_4)$. For better visibility, collinear segments
are slightly shifted vertically, with red and green points indicating overlapping segments.
In an optimal hitting set, the point covering a 
horizontal segment labeled with a variable name $x_i$ induces a truth value for the corresponding variable: 
selecting one of its gray points (e.g., in the indicated green manner) assigns a value of ``true''; selecting the red point at the right end of the segment, a value of ``false''.
Overall, truth assignments for each variable correspond to a set of green or 
red points, respectively.
(Note that there are several equivalent choices from the gray points, which all correspond to the same truth assignments.)
Literals occurring in clauses are indicated by magenta circles; these are the only places
where a point can hit three segments or lines at once.
}
\label{fig:rays_npc}
\end{figure}
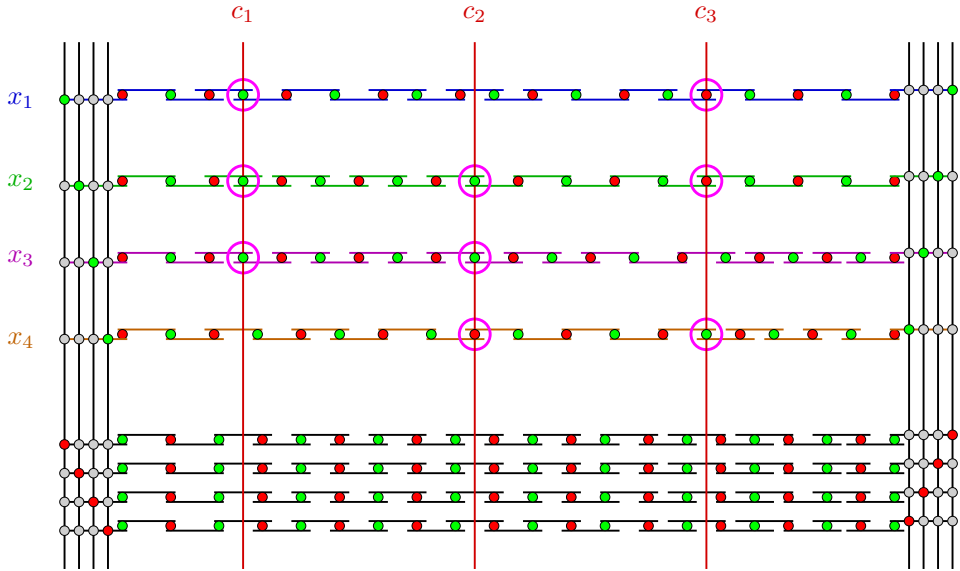

Each variable is represented by a collinear connected set of $2m + 2$
horizontal unit segments: a start segment, a pair of segments for each clause,
and an end segment. Each clause is represented by a red
vertical line that intersects appropriate pairs of horizontal variable
segments (if that variable occurs in a clause) or just single segments
(in case a variable does not occur in a clause).  Setting appropriate
parities for the literals in a clause is achieved by appropriate
horizontal shifting of the segments, as shown in the figure.  This
results in a construction in which the only place where three of the
elements (segments or lines) can be hit involves a vertical line
representing a clause, corresponding to literals occurring in the
respective clauses. (These are indicated by magenta circles in the
figure.)  There are $n$ black vertical lines intersecting each of the variable start segments 
and $n$ black vertical lines intersecting each of the variable end segments.
Let $N_H = 4mn + 4n$ be the number of horizontal segments.  This includes
the $2m + 2$ per variable just described, and $2m + 2$ more per variable at the bottom
of the instance as shown in Figure~\ref{fig:rays_npc}.  Let $N = N_H + 2n$ be the number of
horizontal segments plus the black vertical lines.

We show that any feasible hitting set with exactly $N/2$ points
induces a truth assignment and vice versa.  The vertical black lines are parallel, so no point
can hit more than one of them. There is no point that
hits more than two of the horizontal segments at once.  There is also
no point on a black vertical line that hits more than one horizontal segment. Therefore,
stabbing all $N$ objects requires at least $N/2$ points, and any
solution consisting of exactly $N/2$ points must hit each object
(horizontal segment or vertical black line) exactly once and hit two objects.
We now argue that hitting the $N$ objects with exactly $N/2$ points induces a truth assignment.
For ease of exposition, call a set of colinear horizontal segments a row.  There are $2n$ rows: $n$ variables rows and $n$ bottom rows.
Consider the start segment of a top row.  That segment will
be hit in one of two ways.  If it shares a hit point with the next horizontal segment
(such as the point colored red on the segment next to $x_1$ in Figure~\ref{fig:rays_npc}), then the variable is set to false.  If it
shares a hit point with one of the black vertical lines (such as the point colored green next to $x_1$), then the variable is set to true.
Suppose there are $q$ variables set to false.  Then there are $q$ black vertical lines that were not hit with variable start segments.
They must be hit by bottom row start segments.  Arbitrarily match each false variable row one-to-one with these $q$ bottom rows.  Each pair
of variable row and matching bottom row 
corresponds to a loop where all selected points are red (in Figure~\ref{fig:rays_npc}).  This leaves $n-q$ true variables.  Since they
collectively hit $n-q$ of the left black vertical lines, then there are $n-q$ bottom start segments that share hit points with the next
segment and share no hit point with a black vertical line.  Similarly match the true variables with these $n-q$ bottom rows to form $n-q$ loops covered only by green points.
Thus, any solution of size $N/2$ hitting the variable components must select all red or all green
points from each variable's loop, corresponding to a truth assignment.
We get an overall feasible hitting set if and only if the points also
stab the vertical clause lines, corresponding to a satisfying truth
assignment.

Any satisfying assignment can hit all segments and lines with $N/2$ points by setting the truth variable loops as described above.  A satisfying
assignment will also hit every clause line.
  
\qed
\end{proof}

After appropriate vertical scaling, we can replace the vertical lines by vertical unit segments,
immediately giving the following corollary. 

\begin{corollary}
\label{cor:unit_npc}
Deciding if there exists a set of $k$ points in the plane that hit a
given set $S$ of unit-length axis-parallel segments is NP-complete.
\end{corollary}

We now show APX-hardness for the all-segment case.

\begin{theorem}
  Computing a minimum hitting set of axis-parallel segments is APX-hard.
  \label{theorem:apx-hard}
\end{theorem}

\begin{proof}
 We give a reduction from MAX-2SAT(3), maximum 2-satisfiability in
 which each variable appears in at most three clauses. MAX-2SAT(3) is
 known to be APX-hard~\cite{ausiello2012complexity}. In our reduction,
 a clause is represented by a vertical segment. A variable gadget is a
 ``loop'' consisting of at most 8 horizontal segments and exactly 2
 vertical segments at the far left/right, linking a chain of an odd
 number (3, 5, or 7, depending if the variable appears in 1, 2, or 3
 clauses) of collinear horizontal segments on the upper portion of the
 gadget to a single horizontal segment closing the loop along the
 bottom portion of the gadget. Refer to Figure~\ref{fig:max2sat3}.  In
 total, a variable loop consists of an even number (6, 8, or 10) of
 segments, whose intersection graph is an even cycle (no three of them
 intersect).  We place red and green points, each representing an edge
 of the cycle that is the intersection graph, alternating around the
 cycle. These green/red points occur at crossings with the left/right
 vertical segments of the loop, or at overlap points along the top
 chain of horizontal segments of the loop.  Membership of a variable
 $x_i$ in a clause $c_j$ is represented by having the clause segment
 pass through a green or red point (according to whether the variable
 or its negation appears in the clause) along the variable loop for
 $x_i$, creating a 3-intersection point at the crossing.

  \begin{figure}[htpb]
    \begin{center}
      \includegraphics[width=0.8\textwidth]{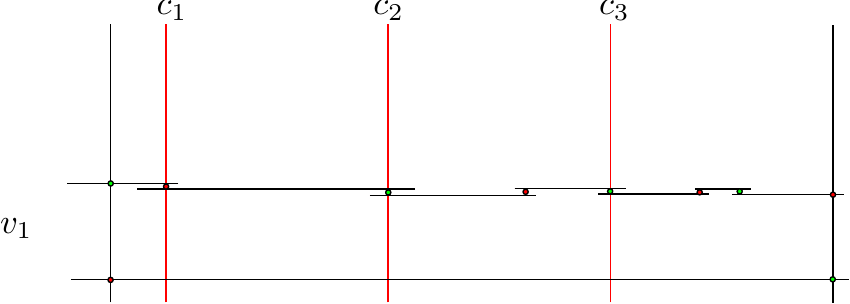}
    \end{center}
    \caption{In this example, clause $c_1$ includes the literal $\overline{v_1}$; clauses $c_2$ and $c_3$ each include the literal $v_1$. }
    \label{fig:max2sat3}
  \end{figure}

Let $m$ and $n$ be the numbers of clauses and variables, respectively,
in an instance of MAX-2SAT(3). Then, $\frac{n}{2} \leq m \leq \frac{3n}{2}$.
Let $k$ be the number of points in an optimal solution of the
corresponding hitting set problem.  If the hitting set does not
contain any 3-intersection, we know that $k \leq 5n + m\leq 11m$, since
all of the segments in each variable loop can be hit using at most 5 hit points.

Suppose ${\cal A}$ is an approximation algorithm for the minimum
hitting point problem on axis-parallel line segments, and that ${\cal
  A}$ guarantees an approximation factor of $1+\varepsilon$. For any
hitting set $H$ (of size $|H|\leq (1+\varepsilon)k$) produced by
${\cal A}$, we obtain a solution for the corresponding MAX-2SAT(3)
instance, as follows.

Consider the variable loop for $x_i$.  In $H$, the number, $h_i$, of
3-intersection hit points along the variable loop could be 0, 1, 2,
or~3.

If $h_i=0$, then we replace the hit points of $H$ along the loop with
an optimal set of hit points along the loop -- either the green or the
red points.  This sets the truth value of $x_i$ (green is ``true'',
red is ``false'').  Further, this exchange has not caused the number
of hit points to go up.

If $h_i=1$, with one 3-intersection (green or red) point $p$, we
replace the hit points of $H$ along the loop with the (optimal) set of
all green or all red hit points along the loop.  This sets the truth
value of $x_i$ (green is ``true'', red is ``false''), and this
exchange has not caused the number of hit points to go up.

If $h_i=2$, then the two 3-intersection points might ``agree'' (be the
same color) or ``disagree'' (be different colors).  If they agree, we
set the truth value of the variable accordingly, and use an optimal
set of hit points of the appropriate color along the loop.  If they
disagree, then we know that the number of hit points of $H$ used in
hitting the variable loop is suboptimal (by at least 1), since $H$
does not use all-red or all-green hit points.  Thus, we set the
variable either way (use an optimal hitting set of all-green or
all-red), and we have a leftover point of $H$, which we use to hit the
clause line that became unhit in the process of setting the hit point
set to be monochromatic.  There was no increase in the number of hit
points.

If $h_i=3$, then if the three points all agree (are of the same
color), we set the truth value of $x_i$ accordingly.  Otherwise, we
know that the set of points of $H$ used to hit segments in this
variable loop is suboptimal; we set the truth value of $x_i$ according
to the majority color among the three 3-intersections, and use the one
saved hit point to hit the clause that was previously hit by the
(minority color) point of $H$, but now is not.

In this way, we have now transformed $H$ into a set $H_v\cup H_c$,
with $|H_v\cup H_c|=|H|$, where $H_v$ is an optimal hitting set for
variable loops (using hit points of a single color around each loop),
plus a (disjoint) set $H_c$ of additional points to hit clause
segments.  Let $alg$ be the number of clauses satisfied by the
variable setting determined by $H_v$.  Then, we know that $alg\geq
m-|H_c|$.

Given an optimal truth assignment for the MAX-2SAT(3) instance,
achieving $opt$ satisfied clauses, one way to construct a hitting set for
all of the segments in the construction is the following: Optimally
place hitting points within each variable loop, according to the truth
assignment (and using exactly $|H_v|$ hit points), and then hit the
remaining $m-opt$ clause lines, not yet hit by 3-intersections within
variable loops, using $m-opt$ separate hit points.  Since $k$ is the
optimal number of hit points for the whole construction, we know that
$k\leq |H_v| + (m-opt)$.

Since we are assuming that ${\cal A}$ is a $(1+\varepsilon)$-approximation,
we have that 
$$|H|=|H_v|+|H_c| \leq (1+\varepsilon)k, $$
which implies that 
$$|H_c| \leq (1+\varepsilon)k - |H_v|.$$
Then, putting these together, and using the facts that 
$opt \geq \frac{m}{3}$ and that $k\leq 11m$, we have
$$alg \geq m-|H_c| \geq m- (1+\varepsilon)k + |H_v|$$
$$ \geq m- (1+\varepsilon)k + k-m + opt \geq (1-33\varepsilon)opt.$$
This implies that our hitting set problem is APX-hard
(a PTAS for the minimum hitting set of axis-parallel
segments would imply a PTAS for MAX-2SAT(3)).
\qed
\end{proof}

\subsection{Approximation}

We give a $5/3$-approximation for hitting a set $V$ of vertical lines and a set $H$ of  horizontal segments.
We start by looking at the lower bounds: $v=|V|$ is the number of
vertical lines. It is a lower bound.  Let $h$ be the lower bound on
hitting horizontal segments only.  We can compute $h$ and a corresponding solution exactly in polynomial time; it is the
minimum number of hit points for the horizontal segments (computed on
each horizontal line).  This is equivalent to hitting a collection of intervals with a minimum number of points and can be solved in polynomial time by a well-known ``folklore'' result, as mentioned in Section~\ref{sec:seg}.  At any stage of the algorithm, we let $h$ and $v$ be the current
values of these lower bounds for hitting the current (remaining unhit) sets $H$ and $V$.

In Stage 1, we place two kinds of points:

(a) We place hitting points on vertical lines that reduce $h$ (and $v$) by one.
These points are ``maximally productive'' since no single hitting point can do more than to reduce $h$ and $v$ each by one.
As vertical lines are hit, we remove them from $V$.
Similarly, as horizontal segments are hit, we remove them from~$H$.

(b) Look for pairs (if any) of points, on the same horizontal line and
on two vertical lines (from among the current set $V$), that
decrease $h$ by one.

Let $k_1$ and $k_2$ be the number of Type (a) and Type (b) points placed in this stage, respectively.
Therefore, for the remaining instance, the lower bound $h$ decreases by $k_1 + k_2/2$, and $v$ decreases by $k_1 + k_2$.

In Stage 2, we now have a set of vertical lines $V$ and horizontal segments $H$ such that no
single point at the intersection of a vertical line and a horizontal segment (or segments)
reduces $h$, and no pair of points on two distinct vertical lines 
reduces $h$.

\begin{lemma}
  \label{lem:stage3-lb}
For such sets $V$ and $H$ as in Stage 2, an optimal hitting set has
size at least $v+h$, where $v=|V|$ and $h$ is the minimum number of
points to hit $H$.
\end{lemma}

\begin{proof}
The hit points we place on $V$ (one per line) might conceivably decrease $h$.
We claim that this cannot happen.
Assume to the contrary that it happens.
Let $\{q_1, q_2,\ldots, q_K\}$ be a minimum-cardinality set such that each of them is on some line of $V$ from left to right and $h$ is decreased after placing the set. Since the set is minimum, the points in it should be on a horizontal line $L$. 

Since we have found all productive points and pairs of points in stage 1, $K$ should be at least $3$.
Consider the hit point $q_2$. The segments on $L$ that are
not hit by $q_2$ are either completely left or right of $q_2$;   
let $H_l$
and $H_r$ be the corresponding sets.  Points to the left of $q_2$ 
do not hit $H_r$, and points to the right of $q_2$ do not hit $H_l$.
If adding $q_1$ decreases $H$, that means $q_1$ and $q_2$ is a productive pair, which should be found in stage 1;
otherwise this means that the point $q_1$ is unnecessary, contradicting the minimality of $K$. 
\qed
\end{proof}

The above lemma implies that for Stage 2 it suffices to select one point to hit each unhit vertical line and to independently find an optimal solution for hitting only the unhit horizontal segments.  As mentioned above, the latter can be solved in polynomial time.

\begin{theorem}\label{thm:5/3-approx-vlhs}
There is a polynomial-time $5/3$-approximation algorithm for geometric hitting set
for a set of vertical lines and horizontal segments.
\end{theorem}

\begin{proof}
Let $v$ be the total number of vertical lines in the instance and $h$ be the minimum number of points required to hit only the horizontal segments in the instance.  The total number of points selected by our algorithm is $k_1 + k_2$ from the first stage and $h - k_1 - k_2/2 + v - k_1 - k_2$ from the second stage.  By Lemma~\ref{lem:stage3-lb}, the number of points chosen in Stage 2 is a lower bound on the cost of an optimal solution:
\begin{equation}\label{eq:stage3-bound}
h - k_1 - k_2/2 + v - k_1 - k_2 \leq OPT.
\end{equation}
We also have $h \leq OPT$ and $v \leq OPT$.  There are two cases:
\begin{enumerate}
\item[(i)] $k_1 + k_2 \leq 2/3\cdot OPT$: In this case we select at most $2/3\cdot OPT$ points in Stage 1, and we use~\eqref{eq:stage3-bound} to bound the number of points selected in Stage 2.  We conclude that our algorithm selects at most $5/3\cdot OPT$ points.
\item[(ii)] $k_1 + k_2 > 2/3 \cdot OPT$: The total number of points selected by our algorithm is  $h - k_1 - k_2/2 + v \leq 2\cdot OPT - (k_1 + k_2/2)$.  Since $k_1 + k_2/2 \geq k_1/2 + k_2/2 > 1/3 \cdot OPT$, we obtain a 5/3-approximation in this case as well.
\end{enumerate}
\qed
\end{proof}

\begin{theorem}
There is a polynomial-time $5/3$-approximation algorithm for geometric hitting set
for a set of vertical (downward) rays and horizontal segments.
\end{theorem}
\begin{proof}
The 2-stage approximation algorithm described above works for this case as well. The key observation 
is that among any set of collinear downward rays, we may remove all but the one with the lowest apex from the instance, and we obtain a proof for this case by replacing ``line'' with ``ray'' in the proof above.
\qed
\end{proof}

\section{Hitting Pairs of Segments}
\label{sec:pai}
We consider now the hitting set problem for inputs that are {\em unions} of two
segments, one horizontal and one vertical.  While we are motivated by pairs (and larger sets) of segments that 
form paths, our methods apply to general pairs of segments, which might
meet to form an ``L'' shape, a ``$+$'', or a ``T'' shape, or they may
be disjoint.
This hitting set problem is NP-hard, since it generalizes the
case of horizontal and vertical segments.

\begin{theorem}\label{thm:4-approx-pair-vlhs}
For objects that are unions of a horizontal and a vertical segment,
the hitting set problem has a polynomial-time 4-approximation.
\end{theorem}

\begin{proof}
For ease of discussion, we call the union of two segments an ``L''.
We use a method similar to those used
in~\cite{carr2000approximation,gaur2000constant}. 

Briefly, we do the following. Solve the natural set-cover linear
programming (LP) relaxation.  Create two new problems: one that has
only the horizontal piece of some of the Ls and another that has only
the vertical pieces of the remaining Ls.  Place an L into the vertical
problem if the LP vertical segment has value at least $1/2$, and into
the horizontal problem otherwise. Solve the two new problems in
polynomial time using the combinatorial method for the 1D problem, or
solving the LPs, which are totally unimodular, and thus will return
integer solutions.  Take all the points selected by either new
problem.  We prove that these points are a $4$-approximation.

In more detail, suppose we have $l$ unions of segments as described
above, and let $P$ be the set of points serving as our potential
hitters.  We assume that $|P|$ is polynomial in $l$ by preprocessing
the instance, if necessary, so that we only consider points at
endpoints and crossings of segments.  For each such union $i$, we let
$S_i$ be the set of points covering the union, while $H_i$ and $V_i$
are the sets of points covering the horizontal segment and vertical
segment respectively.  We employ the standard set cover linear program
(LP) relaxation specialized to our problem:
\begin{alignat}{3}
\notag \text{min } \sum_{p \in P} x_p\\
\label{eq:LP-cover} \sum_{p \in H_i} x_p + \sum_{q \in V_i} x_q &\geq 1, \ \forall \ 1 \leq i \leq l\\
\notag 0 \leq x_p &\leq 1, \ \forall\ p \in P.
\end{alignat}

We use an optimal LP solution, $x^*$, to construct a new instance of
the problem in which each union contains either a vertical segment or
a horizontal segment, but not both.  This new instance is easier to
approximate but no longer provides a lower bound on the original
optimum value, $OPT$; however, we show that it provides a lower bound that
is within a constant factor of $OPT$.

For each union of segments $i$, we set 
\begin{equation}\label{eq:LP-filtering}
S'_i =
\begin{cases}
H_i,\ \text{if } \sum_{p \in H_i} x^*_p \geq 1/2\\ 
V_i,\ \text{otherwise}.
\end{cases}
\end{equation}
Now each $S'_i$ corresponds to either a horizontal or vertical
segment.  Let $H' = \{i \mid S'_i \text{ represents a horizontal
segment}\}$, and let $V' = \{1,\ldots, l\} \setminus H'$.  Our
algorithm is as follows:
\begin{enumerate}
\item Solve the LP, and let $x^*$ be an optimal solution.
\item Construct $S'_i$ for each union of segments $i$ as described above. 
\item Solve the hitting set problems for all the horizontal segments, $H'$, 
and all the vertical segments, $V'$, independently.  Return the union
of the points, $X$ selected by optimal solutions to each instance.
\end{enumerate}

This algorithm returns a feasible solution since it selects some point
in $S'_i \subseteq S_i$ for each union of segments.  The first two
steps run in polynomial time. Hitting segments of a single orientation
is solvable in polynomial time; in fact the corresponding set cover LP
relaxation in this case has the consecutive ones property and is
totally unimodular, hence the optimum LP value equals the optimum
integer solution value.

To see that it is a 4-approximation, let $y^*_p = \min\{2x_p^*,\ 1\}$
for all $p$.  By \eqref{eq:LP-cover} and \eqref{eq:LP-filtering} we
see that the fractional vector $y^*$ is feasible for the LP instance
defined by the segments corresponding to the $S'_i$.  Now we modify
the latter LP instance by taking each point $p$ and replacing the
variable $x_p$ with variables $x_{p,h}$ and $x_{p,v}$, where $x_{p,h}$
appears only in horizontal segment constraints where $x_p$ formerly
appeared, and $x_{p,v}$ appears only in such vertical segment
constraints.  The resulting LP decouples the horizontal and vertical
segments and captures precisely the problem from Step 3 of the
algorithm.  Since this LP is totally unimodular, we have that the
number of chosen points, $|X|$, is at most the cost of any feasible
fractional solution.  In particular we see that the fractional vector
$z^*$ with $z^*_{p,h} = z^*_{p,v} = y^*_p$ is feasible for the
decoupled LP, and so:
\begin{equation}\label{eq:decoupled-LP-bound}
|X| \leq \sum_{q} z^*_q = 2\sum_{p} y^*_p.
\end{equation} 
To obtain our desired result we note that $\sum_{p} y^*_p \leq
2\sum_{p}x^*_p$ by the definition of $y^*$, yielding $|X| \leq
4\sum_{p} x^*_p \leq 4\cdot OPT$ by \eqref{eq:decoupled-LP-bound}.
\qed
\end{proof}
The above idea naturally extends to a 4-approximation for the weighted
version of the problem.  For unions consisting of at most $k$ segments
drawn from $r$ orientations, the approach yields a $(k\cdot
r)$-approximation.

The LP-rounding technique in the proof above was introduced by Carr et
al.~\cite{carr2000approximation} to obtain a 2.1-approximation for the
weighted edge-dominating set problem.  A similar idea was introduced
independently by Gaur et al.~\cite{gaur2000constant} to obtain a
2-approximation for stabbing axis-aligned rectangles with horizontal
and vertical lines.  By using the approach above in conjunction with
our approximation algorithm for Theorem~\ref{thm:5/3-approx-vlhs}, we
obtain an improved approximation factor in the case that the vertical
segments are lines.  Before describing this result, we need a slightly
stronger version of Theorem~\ref{thm:5/3-approx-vlhs}:

\begin{lemma}\label{lem:5/3-LP-approx-vlhs}
There is a polynomial-time 5/3-approximation algorithm for hitting a
set of vertical lines and horizontal segments that always returns a
solution of cost within 5/3 that of an optimal solution to the natural
set cover LP relaxation.
\end{lemma}

\begin{proof}
Given an instance of geometric hitting set over vertical lines and
horizontal segments, let $LP^*$ be the optimum value achieved by the
natural set cover LP relaxation.  We show that the algorithm used to
establish Theorem~\ref{thm:5/3-approx-vlhs} satisfies stronger
versions of the bounds used in the proof of
Theorem~\ref{thm:5/3-approx-vlhs}:
$$
h \leq LP^*,\, v \leq LP^*,\,\text{and } h - k_1 - k_2/2 + v - k_1 -
k_2 \leq LP^*.
$$
Since the vertical lines are disjoint, by summing the corresponding LP
constraints, we see that $\sum_p x_p \geq v$ for any feasible $x$.
Taking $x$ to be an optimal solution, $x^*$, we have that $LP^*
= \sum_p x^*_p \geq v$.  As noted before, the natural set cover LP
relaxation is totally unimodular in the case of hitting only
horizontal segments.  Thus, by dropping the constraints corresponding
to the lines from the LP, we conclude that $LP^* \geq h$.

For the final bound, we need to show that for the type of instance
obtained by our 5/3-approximation in Stage 2, $LP^*$ is equal to
$OPT'$, the optimum size of a hitting set.  Lemma~\ref{lem:stage3-lb}
shows that for such instances, $OPT' = v' + h'$, where $v'$ and $h'$
are the individual vertical and horizontal lower bounds for the
instance.

Consider a collection of collinear horizontal segments from a Stage-2
instance, and remove all points that lie on some vertical line along
with all the horizontal segments hit by such points.  The proof of
Lemma~\ref{lem:stage3-lb} shows that such a deletion does not increase
the optimal number of points required to hit such an instance.  Hence,
appealing to the integrality of such LP instances when dropping the
vertical line constraints, we have that $\sum_{p \in P' \setminus
P'_V} x_p \geq h'$, where $P'$ is the set of points of a Stage-2
instance, and $P'_V$ is the set of points that lie on some vertical
line.  Considering only the vertical line constraints, as above, gives
us $\sum_{p \in P'_V } x_p \geq v'$.  Together, these inequalities
yield the desired bound, $\sum_{p \in P'} x_p \geq h' + v'$.

We substitute these bounds in the proof of
Theorem~\ref{thm:5/3-approx-vlhs} to conclude that our algorithm
selects at most $5/3\cdot LP^*$ points instead of $5/3\cdot OPT$
points.
\qed
\end{proof}

Using similar methods and the above lemma, we also have the following:

\begin{theorem}
\label{thm:8}
For objects that are unions of a horizontal segment and a vertical
line, the hitting set problem has a polynomial-time
10/3-approximation.
\end{theorem}

\begin{proof}
Our algorithm is essentially the same as the 4-approximation of Theorem~\ref{thm:4-approx-pair-vlhs}, with a different last step:
\begin{enumerate}
\item Solve the LP, and let $x^*$ be an optimal solution.
\item Construct $S'_i$ for each union $i$ of a horizontal segment and a vertical line. 
\item Now each $S'_i$ is either a horizontal segment or a vertical line, and we find a feasible solution $X$ for this instance using our 5/3-approximation. 
\end{enumerate}

We construct $y^*$ just as in the proof of our 4-approximation; however, now we observe that $y^*$ is feasible for the set cover LP relaxation for the instance defined by the $S'_i$.  This is just an instance of hitting horizontal segments and vertical lines, and so:
$$|X| \leq 5/3\cdot LP^* \leq 5/3\cdot \sum_p y^*_p.$$ 
Since $\sum_{p} y^*_p \leq 2\sum_{p}x^*_p$ as before, we have that $|X| \leq 10/3 \cdot \sum_p x^*_p \leq 10/3 \cdot OPT$ as desired.
\qed
\end{proof}

\section{Hitting Triangle-Free Sets of Segments}
\label{sec:triangle-free}
We consider now the problem in which the $n$ input segments $S$ are
allowed to cross or to share endpoints, but not to overlap (i.e., the
intersection of any two input segments is not a non-zero length
segment -- it is either empty or a single point).

Let $G=(V,E)$ denote the (planar) {\em arrangement graph}, $G(S)$, induced by the
segments $S$; thus, $G$ has vertex set $V$ equal to the set of all
endpoints or crossing points of $S$ and has edge set $E$ of $m=|E|$ edges joining each
pair of vertices that appear consecutively along a segment of~$S$.

We assume that $G$ is {\em triangle-free}, meaning that it has no
cycle of length 3 (i.e., its {\em girth} is at least 4).  It is well
known that a planar triangle-free graphs must have a vertex of degree
at most 3.  (For completeness, we provide the proof: In a
triangle-free planar graph (having $n$ nodes, $e$ edges, and $f$
faces), each face has at least 4 edges bounding it.  The sum of the
number of edges bounding each of the faces is simply $2e$, and in a
triangle-free graph must be at least $4f$; thus, $2e\geq 4f$.  By
Euler's formula ($f-e+n=1+c$, for $c\geq 1$ connected components), we
get $2e\geq 4(1+c+e-n)\geq 4(2+e-n)$, implying that $e\leq 2n-4$.  The
sum of the vertex degrees is exactly $2e$ and is thus at most $4n-8$;
thus, not all vertices have degree 4 or more -- there must be a vertex
of degree at most~3.)

In this section we give a linear-time 3-approximation algorithm for
computing a hitting set of points that hit all of the segments of $S$,
assuming that the arrangement graph $G(S)$ is triangle-free and given.  (If $G(S)$ is not
given, we can compute $G(S)$ from $S$ in time $O(m+ n\log n)$, using, e.g., the algorithm of Balaban~\cite{b-oafsi-95}.)  Our
approximation factor of 3 matches that obtained recently by Joshi and
Narayanaswamy~\cite{joshi2014approximation}; however, their algorithm
employs linear programming, while ours is a simple, combinatorial
linear-time ($O(m)$) algorithm.

Our algorithm is the following clipping/shortening process: 

\begin{description}
\item[(i)] 
Pick a vertex $v\in V$ of degree at most 3 (it will necessarily
be a segment endpoint); such a vertex must exist, by the triangle-free
property.

\item[(ii)] 
Remove the vertex $v$, and shrink the incident segments with
endpoint $v$ to the next adjacent vertex.
(In particular, if $v$ is a T-junction, where two of the edges
incident to $v$ lie on a common segment, then only the one segment
with endpoint at $v$ is shrunk, leaving the other two edges
connected.)

\item[(iii)]  
When shortening a segment $s$ results in segment $s$ becoming
a single point (vertex), $u$, establish a hitting point at $u$ and
remove all segments that pass through~$u$.
\end{description}

The following invariants hold at any stage of the process:

\begin{description}
\item[(1)] 
There is at most one remaining subsegment of an input segment
(i.e., the portion of an original segment $s$ that remains is
connected).

\item[(2)] 
All segments that have been removed are hit by the hitting points
that have been established.

\item[(3)] 
Any hitting set of the remaining segments, together with the
established hitting points already found, forms a hitting set for the
original set of input segments.

\item[(4)] 
The graph $G$ remains triangle-free during the process.
\end{description}

The invariants imply that the set of points computed by the algorithm
is a valid hitting set.  The following lemma establishes the
approximation factor:

\begin{lemma}
The number of hitting points established by the algorithm is at
most 3 times the number, $|H|$, of points in any hitting set $H$ for
$S$.
\end{lemma}

\begin{proof}
Place tokens on the vertices $H$ and consider running the
clipping/shortening process on $G$, with the following actions on the
tokens.

When there is a token on the vertex $v$ that is about to be clipped,
replace the token with at most 3 clones of it, one on each of the
segments that meet at $v$, allowing each clone to slide along
with the endpoint of a clipped segment $s$ as the segment is shrunk,
leaving the clone at a new vertex $u$, the new endpoint of
segment $s$. 
(There might also be a token at $u$ already; we allow two or more tokens/clones
to accumulate at a vertex.)
We never clone a clone; if a clone associated with a segment $s$ 
exists at a vertex $v$ that is
being clipped, it remains on segment $s$, and slides along it as it shrinks.  
Thus, associated with each point of $H$ there is either a single token
or up to 3 clones of the token (but not both).

This ensures that the tokens/clones continue to hit all segments, at
all stages of the clipping/shortening process.  
(Here, we are using the degree-3 property, which allows us to make sure that
two edges incident on $v$ that lie on the same segment $s$ are not cut apart at
$v$ in our process; thus, a point of $H$ that lies on $s$ continues to hit the 
shrunk version of segment $s$.  If we had split $s$ at $v$, with no point of $H$ at $v$,
then no clones are generated at $v$, and the point(s) of $H$ on segment $s$ may no longer
be a valid hitting set for the new arrangement after splitting $s$ at $v$.)
In particular, when a segment shrinks to a point $u$, there is at
least one token/clone present there.  Thus, the number of hitting
points established by our algorithm is at most $3|H|$, for any hitting
set $H$ of $S$.  Letting $H$ be an optimal hitting set, we get that
the number of hitting points produced by the algorithm is at most 3
times optimal.  \qed
\end{proof}

\begin{theorem} \label{thm:3-approx}
The algorithm yields a 3-approximation and runs in time
$O(m)$, where $m$ is the number of edges in the original (planar)
arrangement graph $G$.
\end{theorem}

\begin{proof}
Immediate, since we only have to maintain the graph $G$ in a standard
planar network data structure (e.g., the Doubly Connected Edge List
(DCEL)~\cite{bcko}) that allows us to know vertex degrees and perform
elementary operations in constant time.  \qed
\end{proof}

\section{Conclusion}
\label{sec:conclusion}

We have given a variety of new hardness and approximation results for
geometric hitting sets involving lines, rays, and segments from a
small number of discrete orientations.  We have also given a
linear-time combinatorial algorithm that yields a 3-approximation for
hitting triangle-free sets of non-overlapping segments in the plane,
matching the approximation factor recently obtained
by~\cite{joshi2014approximation} using linear programming methods.

We note that our methods apply as well to yield the same results
(lower bounds, approximation bounds) for the more general setting in which
``segments,'' ``rays,'' and ``lines'' are given as subsets of families
of disjoint pseudoline curves, with each disjoint family playing the
role of an ``orientation'' of lines.  (The pseudoline property
requires that any two pseudoline curves that intersect (are not
``parallel'') do so in a single point of intersection, where they
cross.)

Natural open questions ask if any of these approximation bounds can be
improved.  Notably, we believe that the trivial 2-approximation for
hitting segments of two orientations can be improved.  Another
direction for future research is fixed parameter tractable (FPT)
algorithms; for some recent related work, see \cite{DBLP:conf/compgeom/AfshaniBDN16}.  Finally, we are interested in optimal coverage versions
of these problems in which, e.g., one desires a smallest cardinality
set of line segments, rays, or lines, from a small number of
orientations, in order to cover a given set of points.

\begin{acknowledgements}
This work is supported by the Laboratory Directed Research and
Development program at Sandia National Laboratories, a multi-program
laboratory managed and operated by Sandia Corporation, a wholly owned
subsidiary of Lockheed Martin Corporation, for the U.S. Department of
Energy's National Nuclear Security Administration under contract
DE-AC04-94AL85000.
J. Mitchell acknowledges support from the US-Israel Binational Science Foundation (grant 2010074) and the National
Science Foundation (CCF-1018388, CCF-1526406).
\end{acknowledgements}


\begin{thebibliography}{10}
\providecommand{\url}[1]{{#1}}
\providecommand{\urlprefix}{URL }
\expandafter\ifx\csname urlstyle\endcsname\relax
  \providecommand{\doi}[1]{DOI~\discretionary{}{}{}#1}\else
  \providecommand{\doi}{DOI~\discretionary{}{}{}\begingroup
  \urlstyle{rm}\Url}\fi

\bibitem{DBLP:conf/compgeom/AfshaniBDN16}
Afshani, P., Berglin, E., van Duijn, I., Nielsen, J.S.: Applications of
  incidence bounds in point covering problems.
\newblock In: Proc. 32nd International Symposium on
  Computational Geometry, 
  \emph{LIPIcs}, vol.~51, pp. 60:1--60:15. Schloss Dagstuhl - Leibniz-Zentrum
  fuer Informatik (2016).

\bibitem{alon2012non}
Alon, N.: A non-linear lower bound for planar epsilon-nets.
\newblock Discrete \& Computational Geometry \textbf{47}(2), 235--244 (2012)

\bibitem{aronov2010small}
Aronov, B., Ezra, E., Sharir, M.: Small-size $\varepsilon$-nets for
  axis-parallel rectangles and boxes.
\newblock SIAM Journal on Computing \textbf{39}(7), 3248--3282 (2010)

\bibitem{ausiello2012complexity}
Ausiello, G., Crescenzi, P., Gambosi, G., Kann, V., Marchetti-Spaccamela, A.,
  Protasi, M.: Complexity and approximation: Combinatorial optimization
  problems and their approximability properties.
\newblock Springer Science \& Business Media (2012)

\bibitem{b-oafsi-95}
Balaban, I.J.: An optimal algorithm for finding segment intersections.
\newblock In: Proc. 11th Symposium on Computational Geometry, pp. 211--219. ACM (1995)

\bibitem{bcko}
de~Berg, M., Cheong, O., van Kreveld, M., Overmars, M.: Computational Geometry:
  Algorithms and Applications, 3rd edn.
\newblock Springer-Verlag, Berlin, Germany (2008)

\bibitem{brimkov2013approximability}
Brimkov, V.E.: Approximability issues of guarding a set of segments.
\newblock International Journal of Computer Mathematics \textbf{90}(8),
  1653--1667 (2013)

\bibitem{brimkov2010experimental}
Brimkov, V.E., Leach, A., Mastroianni, M., Wu, J.: Experimental study on
  approximation algorithms for guarding sets of line segments.
\newblock In: Advances in Visual Computing, pp. 592--601. Springer (2010)

\bibitem{brimkov2011guarding}
Brimkov, V.E., Leach, A., Mastroianni, M., Wu, J.: Guarding a set of line
  segments in the plane.
\newblock Theoretical Computer Science \textbf{412}(15), 1313--1324 (2011)

\bibitem{brimkov2012approximation}
Brimkov, V.E., Leach, A., Wu, J., Mastroianni, M.: Approximation algorithms for
  a geometric set cover problem.
\newblock Discrete Applied Math \textbf{160}, 1039--1052 (2012)

\bibitem{DBLP:conf/cccg/BrodenHN01}
Brod{\'{e}}n, B., Hammar, M., Nilsson, B.J.: Guarding lines and 2-link polygons
  is {APX}-hard.
\newblock In: Proc. 13th Canadian Conference on Computational Geometry, pp. 45--48 (2001).

\bibitem{bg-aoscf-95}
Br{\"o}nnimann, H., Goodrich, M.T.: Almost optimal set covers in finite
  {VC}-dimension.
\newblock Discrete \& Computational Geometry \textbf{14}, 263--279 (1995)

\bibitem{carr2000approximation}
Carr, R.D., Fujito, T., Konjevod, G., Parekh, O.: A 2 1/10-approximation
  algorithm for a generalization of the weighted edge-dominating set problem.
\newblock In: European Symposium on Algorithms, pp. 132--142. Springer (2000)

\bibitem{chvatal1979greedy}
Chv{\'a}tal, V.: A greedy heuristic for the set-covering problem.
\newblock Mathematics of Operations Research \textbf{4}(3), 233--235 (1979)

\bibitem{clarkson1993algorithms}
Clarkson, K.L.: Algorithms for polytope covering and approximation.
\newblock In: Algorithms and Data Structures, pp. 246--252. Springer (1993)

\bibitem{cv-iaagsc-07}
Clarkson, K.L., Varadarajan, K.: Improved approximation algorithms for
  geometric set cover.
\newblock Discrete \& Computational Geometry \textbf{37}(1), 43--58 (2007)

\bibitem{dinur2014analytical}
Dinur, I., Steurer, D.: Analytical approach to parallel repetition.
\newblock In: Proc. 46th Symposium on Theory of
  Computing, pp. 624--633. ACM (2014).

\bibitem{dom2009parameterized}
Dom, M., Fellows, M.R., Rosamond, F.A.: Parameterized complexity of stabbing
  rectangles and squares in the plane.
\newblock In: WALCOM: Algorithms and Computation, pp. 298--309. Springer (2009)

\bibitem{duh1997approximation}
Duh, R.c., F\"{u}rer, M.: Approximation of {$k$}-set cover by semi-local
  optimization.
\newblock In: Proc. 29th Symposium on Theory of
  Computing, pp. 256--264. ACM (1997).

\bibitem{DBLP:journals/comgeo/DumitrescuJ15}
Dumitrescu, A., Jiang, M.: On the approximability of covering points by lines
  and related problems.
\newblock Computational Geometry: Theory and Applications \textbf{48}(9), 703--717 (2015).

\bibitem{even2008algorithms}
Even, G., Levi, R., Rawitz, D., Schieber, B., Shahar, S.M., Sviridenko, M.:
  Algorithms for capacitated rectangle stabbing and lot sizing with joint
  set-up costs.
\newblock ACM Transactions on Algorithms \textbf{4}(3), 34:1--34:17
  (2008)

\bibitem{even2005hitting}
Even, G., Rawitz, D., Shahar, S.M.: Hitting sets when the {VC}-dimension is
  small.
\newblock Information Processing Letters \textbf{95}(2), 358--362 (2005)

\bibitem{gaur2007covering}
Gaur, D.R., Bhattacharya, B.: Covering points by axis parallel lines.
\newblock In: Proc. 23rd European Workshop on Computational Geometry, pp.
  42--45 (2007)

\bibitem{gaur2000constant}
Gaur, D.R., Ibaraki, T., Krishnamurti, R.: Constant ratio approximation
  algorithms for the rectangle stabbing problem and the rectilinear
  partitioning problem.
\newblock In: Proc. European Symposium on Algorithms, pp. 211--219. Springer (2000)

\bibitem{giannopoulos2013fixed}
Giannopoulos, P., Knauer, C., Rote, G., Werner, D.: Fixed-parameter
  tractability and lower bounds for stabbing problems.
\newblock Computational Geometry: Theory and Applications \textbf{46}, 839--860 (2013)

\bibitem{hassin1991approximation}
Hassin, R., Megiddo, N.: Approximation algorithms for hitting objects with
  straight lines.
\newblock Discrete Applied Mathematics \textbf{30}(1), 29--42 (1991)

\bibitem{heednacram2010np}
Heednacram, A.: The NP-hardness of covering points with lines, paths and tours
  and their tractability with FPT-algorithms.
\newblock Ph.D. thesis, Griffith University (2010)

\bibitem{hm-ascpp-85}
Hochbaum, D.S., Maas, W.: Approximation schemes for covering and packing
  problems in image processing and {VLSI}.
\newblock Journal of the ACM \textbf{32}, 130--136 (1985)

\bibitem{joshi2014approximation}
Joshi, A., Narayanaswamy, N.: Approximation algorithms for hitting
  triangle-free sets of line segments.
\newblock In: Proc. 14th Scandinavian Symposium and Workshops on Algorithm Theory, pp. 357--367. Springer (2014)

\bibitem{kovaleva2006approximation}
Kovaleva, S., Spieksma, F.C.: Approximation algorithms for rectangle stabbing
  and interval stabbing problems.
\newblock SIAM Journal on Discrete Mathematics \textbf{20}(3), 748--768 (2006)

\bibitem{kratsch2014point}
Kratsch, S., Philip, G., Ray, S.: Point line cover: The easy kernel is
  essentially tight.
\newblock In: Proc. 25th ACM-SIAM Symposium on
  Discrete Algorithms, pp. 1596--1606. SIAM (2014)

\bibitem{kumar2000hardness}
Kumar, V.A., Arya, S., Ramesh, H.: Hardness of set cover with intersection 1.
\newblock In: Proc. 27th International Colloquium on Automata, Languages and Programming, pp. 624--635. Springer
  (2000)

\bibitem{langerman2005covering}
Langerman, S., Morin, P.: Covering things with things.
\newblock Discrete \& Computational Geometry \textbf{33}(4), 717--729 (2005)

\bibitem{megiddo1982complexity}
Megiddo, N., Tamir, A.: On the complexity of locating linear facilities in the
  plane.
\newblock Operations Research Letters \textbf{1}(5), 194--197 (1982)

\bibitem{mustafa2010improved}
Mustafa, N.H., Ray, S.: Improved results on geometric hitting set problems.
\newblock Discrete \& Computational Geometry \textbf{44}(4), 883--895 (2010)

\bibitem{o-agta-87}
O'Rourke, J.: Art Gallery Theorems and Algorithms.
\newblock The International Series of Monographs on Computer Science. Oxford
  University Press, New York, NY (1987)

\bibitem{pach2013tight}
Pach, J., Tardos, G.: Tight lower bounds for the size of epsilon-nets.
\newblock Journal of the American Mathematical Society \textbf{26}(3), 645--658
  (2013)

\bibitem{urrutia2000art}
Urrutia, J., et~al.: Art gallery and illumination problems.
\newblock Handbook of Computational Geometry \textbf{1}(1), 973--1027 (2000)

\end{thebibliography}
\end{document}